\documentclass[utf8,11pt,letterpaper]{article}
\usepackage[margin=1in]{geometry}
\usepackage[utf8]{inputenc}
\usepackage{fullpage}
\usepackage{authblk}
\usepackage{amsthm,amsmath,amssymb}
\usepackage{hyperref}
\usepackage{thm-restate}

\usepackage[obeyFinal,colorinlistoftodos,textsize=tiny,disable]{todonotes}
\newcommand{\ee}[1]{\todo[color=green!20]{E: #1}}

\newcommand{\?}{?}
\newcommand{\val}{\ensuremath{\mathrm{val}}}
\newcommand{\dom}{\ensuremath{\mathrm{Dom}}}
\newcommand{\fd}{\ensuremath{\mathrm{FD}}}
\newcommand{\pd}{\ensuremath{\mathrm{PD}}}
\newcommand{\Oh}{\mathcal{O}}

\let\originalleft\left
\let\originalright\right
\renewcommand{\left}{\mathopen{}\mathclose\bgroup\originalleft}
\renewcommand{\right}{\aftergroup\egroup\originalright}
\renewcommand{\epsilon}{\varepsilon}

\newenvironment{cProof}[1][Proof of Claim]{\begin{proof}[#1]}{\end{proof}}

\theoremstyle{plain}
\newtheorem{theorem}{Theorem}
\newtheorem{lemma}[theorem]{Lemma}

\newtheorem{observation}[theorem]{Observation}
\newtheorem{claim}{Claim}[theorem]

\newtheorem{definition}[theorem]{Definition}

\title{EPTAS for $k$-means Clustering of Affine Subspaces}
\author[1]{Eduard Eiben}
\author[2]{Fedor V. Fomin\thanks{Received fundings from Research Council of Norway via the projects
“MULTIVAL” (grant no. 263317)}}
\author[2]{Petr A. Golovach \protect\footnotemark[1]}
\author[2]{Willian Lochet\thanks{Received funding from the European Research Council (ERC) under the European Union's Horizon 2020 research and innovation programme (grant agreement No 819416). \\\protect\includegraphics[width=3cm]{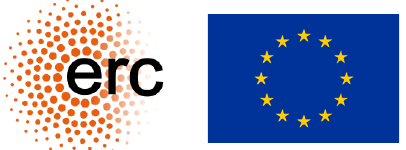}}}
\author[3]{Fahad Panolan
\thanks{Supported by Seed grant, IIT Hyderabad (SG/IITH/F224/2020-21/SG-79).}
}
\author[2]{Kirill Simonov \protect\footnotemark[1]}
\affil[1]{{\small Department of Computer Science, Royal Holloway University of London, Egham, UK}}
\affil[2]{{\small Department of Informatics, University of Bergen, Norway}}
\affil[3]{{\small Department of Computer Science and Engineering, IIT Hyderabad, India}}
\date{} 

\begin{document}
\maketitle

\begin{abstract}
We consider a generalization of the  fundamental $k$-means clustering    for 
 data with incomplete or corrupted entries. When data objects are represented by points in  $\mathbb{R}^d$, a data point is said to be incomplete when some of its entries are missing or unspecified.
An incomplete data point  with at most $\Delta$ unspecified entries corresponds to an axis-parallel affine subspace of dimension at most $\Delta$, called a $\Delta$-point. Thus we seek a partition of $n$ input $\Delta$-points into $k$ clusters minimizing the $k$-means objective. For   $\Delta=0$, when all coordinates of each point are specified, this is the usual $k$-means clustering. 
We give an algorithm that  finds an  $(1+ \epsilon)$-approximate solution in time $f(k,\epsilon, \Delta) \cdot n^2 \cdot  d$ for some function $f$ of $k,\epsilon$, and $\Delta$ only.  

\end{abstract}

\section{Introduction}

Clustering is one of the most widely used techniques in data mining, statistics, and machine learning. 
In general,  the purpose of clustering is to group a set of objects such that similar objects end up in the same cluster. A common approach to clustering is to treat objects with $d$ features as points in $\mathbb{R}^d$. The measure of the similarity between two objects is the Euclidian distance between the corresponding points. 
One of the most famous mathematical models of data clustering is  $k$-means. 
In $k$-means clustering, we want to partition the points in $\mathbb{R}^d$, or some other metric space, by selecting 
 a set of $k$ centers and assign each of the points to its closest center. 
 The quality of the clustering solution is characterized by the $k$-means cost function, which minimizes the sum of squared distances between every point and its nearest center.  
 
It is a common occurrence in practical applications that some features of data can be missing or unspecified. Since missing data could significantly affect the information retrieved from the data, handling such data is a pervasive challenge. 
Various heuristic, greedy, convex optimization, statistical, or even ad hoc methods were proposed throughout the years in different practical domains to handle missing data. We refer to books 
\cite{allison2001missing}, \cite{VidalMS16}, and  the Wikipedia entry\footnote{\url{https://en.wikipedia.org/wiki/Missing_data}} for an introduction to the topic.

 Gao, Langberg, and  Schulman in~\cite{GaoLS08} proposed the following geometric approach to the clustering of incomplete data. 
 A data object that misses $ \Delta $ entries corresponds to a $ \Delta $-dimensional affine subspace in  $\mathbb{R}^d$.  This subspace is parallel to coordinate axes corresponding to the missing coordinates. We call such affine subspaces \emph{$\Delta$-points}.
 %Note that all $\Delta$-points are pairwise-orthogonal.
 With this notation, a regular point in $\mathbb{R}^d$ is a $0$-point.
The distance between a $\Delta$-point $x$ and a point $y$ is naturally defined as the minimum distance between $y$ and a point from $x$.  In this setting, the classical $k$-means and other clustering problems like $k$-median and $k$-center, can be defined on a set of  $ \Delta $-points. The only difference is that we minimize the corresponding objective function based on the distances between the center of the cluster and the $\Delta$-points from the cluster. 
Gao et al.'s geometric model has the following explanation: It provides the values of missing entries that are most suitable for clustering objective. In particular,  under the assumption that the set of ``complete'' data objects is well-clustered, this approach yields the correct clustering.

 From the computational perspective, the $\Delta$-point clustering models are way more challenging than their vanilla clustering counterparts.
Most of the clustering algorithms crucially exploit the fact that clustering occurs in a metric space.  
The  major obstacle of using these algorithms for the more general clustering problem is that 
  the ``distances'' between $\Delta$-points do not satisfy the triangle inequality. 
As    Gao et al.~\cite{GaoLS08} wrote: 
  ``\emph{This problem defeats many existing algorithmic approaches for “clustering”-type tasks, and for good reason---the geometry seems, in a genuine sense, to be absent.}''
  
  While the definition of clustering of Gao et al.~\cite{GaoLS08} is applicable to $k$-center, $k$-means, and $k$-median versions of clustering, most of the work in this direction concentrated  on $k$-center. (In $k$-center clustering the objective is to minimize the distance $r$ such that every point is within distance $r$ from at least one of the $k$ centers.) 
  Only very recently the first approximation algorithm for $k$-means $\Delta$-point clustering was given by Marom and Feldman \cite{MaromF19} for the special case of $\Delta =1$. 
  We  discuss in details the literature  relevant to our work in the next subsection.

The main result of our paper is the following theorem, which is  the first step in the study of the computational complexity of $k$-means $\Delta$-point clustering beyond $\Delta \in\{0,1\}$. 

\begin{restatable}{theorem}{mainTheorem}
	\label{th:kmean_FPTAS}
The problem of $k$-means clustering of $\Delta$-points in $\mathbb{R}^d$ admits an $(1+ \epsilon)$-approximation algorithm with running time ${2^{\Oh(\frac{\Delta^7k^3}{\epsilon}\log\frac{k\Delta}{\epsilon})}} n^2d$.
%$f(k,\epsilon, \Delta) n^2d$, for some function $f$ of $k,\epsilon,$ and  $\Delta$ only.
\end{restatable}

%\begin{theorem}\label{th:kmean_FPTAS}
%    The problem of $k$-means clustering of $\Delta$-points in $\mathbb{R}^d$ admits an $(1+ \epsilon)$-approximation algorithm with running time $f(k,\epsilon, \Delta) n^2d$, for some function $f$ of $k,\epsilon,$ and  $\Delta$ only.
%\end{theorem}
%This is the first algorithm with   proven guarantee performance for $\Delta>1$. 

%\begin{theorem}\label{th:kmean_FPTAS}
%    The problem of $k$-means clustering $\Delta$-points in $\mathbb{H}^d$ admits an $(1+ \epsilon)$-approximation algorithm with running time $\Oh(f(k,\epsilon, \Delta) n^2d)$
%\end{theorem}

\subsection{Related Work}
The study of 
$k$-means clustering dates back to 1980s \cite{Lloyd82}. The problem is NP-hard even for $k=2$~\cite{AloiseDHP09} and development of approximation algorithms for $k$-means is an active research directions for many years~\cite{DBLP:conf/soda/FeldmanSS13,DBLP:conf/stoc/BadoiuHI02,har2004coresets,DBLP:conf/focs/AhmadianNSW17,KumarSS10,ArthurV07,Cohen-Addad18}. From this list of literature,    the paper of Kumar, Sabharwal,  and   Sen~\cite{KumarSS10} is the most relevant to our work. In that paper,  Kumar et al. gave an $(1+\epsilon)$-approximation  algorithm for $k$-means that runs in time $2^{(k/\epsilon)^{\Oh(1)}} n d$.  Theorem~\ref{th:kmean_FPTAS} provides an extension of the algorithm of 
Kumar et al. to clustering of $\Delta$-points.

Clustering of  $\Delta$-points was defined by  Gao, Langberg, and  Schulman~\cite{GaoLS08}. (They call $\Delta$-points axis-parallel $\Delta$-flats.) In ~\cite{GaoLS08} and the consecutive work~\cite{GaoLS10}, Gao et al. developed several constant factor approximation algorithms for $k$-center clustering of $\Delta$-points and lines.  Lee and Schulman  \cite{LeeS13} gave a $(1+\epsilon)$-approximation algorithm for $k$-center $\Delta$-point clustering that runs in time $2^{\Oh(\Delta k \log k(1+1/\epsilon^2))}nd$.  On the negative side, they show that even if one of $k$ or $\Delta$ (but not both) is a fixed constant greater that $3$, it is NP-hard to decide whether there is a $k$-center $\Delta$-point clustering of value $0$. This implies that there is no approximation algorithm running in time polynomial in $n+d+k$, respectively polynomial in $n+d+\Delta$, for any approximation factor for $k$-center as well as $k$-median and $k$-means clustering of $\Delta$-points. Eiben et al. \cite{Eiben19} provide a thorough study of different variants of $k$-center with incomplete information.

A number of results on $k$-means and $k$-median clustering of lines can be found in the literature. 
 Ommer and  Malik \cite{ommer2009multi} studied $k$-median clustering of lines in $\mathbb{R}^3$. Their algorithm does not have any approximation guarantee and can run for unbounded time. 
 Perets~\cite{perets2011clustering} gave an algorithm that in time $n (\log n/\epsilon)^{\Oh(k)} d $ finds a $(1+\epsilon)$-approximate solution for $k$-median line clustering in $\mathbb{R}^2$. 
Finally, Marom and Feldman in \cite{MaromF19} gave  the first PTAS for $k$-means clustering of lines by providing an  $(1+\epsilon)$-approximation algorithm of running time $f(k, d,\epsilon)n\log{n}$.  The algorithm of Marom and Feldman follows from the construction of a coreset of size 
$dk^{\Oh(k)} \log n/\epsilon^2$. Comparing Theorem~\ref{th:kmean_FPTAS} with the result of Marom and Feldman, since every $1$-point is a line,  their result  implies a  
 PTAS for $k$-means clustering of $\Delta$-points for  $\Delta=1$. However, Theorem~\ref{th:kmean_FPTAS} implies PTAS  only for axis-parallel lines. To the best of our knowledge, no approximation algorithm was known for $k$-means for $\Delta>1$. 
 
\subsection{Overview of the Algorithm}
In order to describe our algorithm, let us recall roughly the argument of Kumar et al.~\cite{KumarSS10} for the case when $\Delta = 0$ and $k =2$. Let $P$ denote the set of points in the instance, let $(P_1, P_2)$ be an optimal partition of $P$ such that $|P_1| \geq |P_2|$, and let $(c_1, c_2)$ be the optimal cluster centers for this partition. The algorithm starts by picking at random some $s = s(\epsilon)$ points $S \subseteq P$. Because $|P_1| \geq |P_2|$, it means that with constant probability, all these points belong to $P_1$ and the center $c'_1$ of $S$ gives a good approximation of $c_1$. Once this is achieved, the algorithm tries to sample inside $P_2$ in order to get an approximation of $c_2$. Because $|P_1|$ can be very large compared to $|P_2|$, the algorithm needs to remove some elements of $P$ from the sampling pool. What Kumar et al. show is that, if $t$ denote the distance between $c_1$ and $c_2$, then the ball $B$ of radius $t/4$ around $c'_1$  contains only elements of $P_1$. Moreover, they show that either $P_2$ is large compared to $P_1 - B$ or the solution containing only one cluster with center $c'_1$ is a good enough approximation. Therefore, by guessing an approximation of $B$, the algorithm is able to sample inside $P_2$ with constant probability and thus obtain a good estimate $c'_2$ for $c_2$. 

Let us now explain some of the difficulties encountered while trying to generalize this argument to $\Delta$-points. Let $P$ denote an instance of $2$-clustering with $\Delta$-points and let $(P_1, P_2)$ be an optimal clustering with centers $(c_1, c_2)$. The first problem we encounter is that sampling elements of $P_1$ might not give a good approximation for $c_1$. Indeed, suppose, for example, that for almost all the points in $P_1$, the first coordinate value is missing. In that case, almost surely, a constant number of randomly sampled elements of $P_1$ will be such that for all of them the value of the first coordinate is missing,
and we get no information about the value of $c_1$ on this coordinate. However, we can show that the number of such ``bad'' coordinates is at most $ \Delta $. This means that we can obtain a good approximation of $c_1$ on some set of coordinates $I_1$ such that $|I_1| \geq d - \Delta$. Let us call this approximation $u_1$. Moreover, a large portion of the elements of $P_1$ will have all their missing values outside of $I_1$. For these points, $u_1$ contains all the information necessary to decide whether they belong to $P_1$ or not. So the algorithm will then guess these points and remove them from the sampling pool (we will come back to how exactly this is done later).
Afterward, either $P_2$ is large enough so that we can sample inside this set with good probability, or
what remains of $P_1$ is still larger than $P_2$. In the last case, we can sample in $P_1$ to obtain information on the value of $c_1$ outside of $I_1$.
The first time we sample inside $P_1$ we get information about $d - \Delta$ coordinates, and with each of the subsequent set of samples, learn at least one new coordinate. Hence we only have to do $2 (\Delta +1)$ sampling steps to obtain an estimate of all the coordinates of $c_1$ and $c_2$.

However, the major problem that we face is how to find the elements of $P_1$ with missing values outside of $I_1$ to remove from the sampling pool. The triangular inequality does not hold for $\Delta$-points, and in particular, if $t$ is the distance between $c_1$ and $c_2$, it is not true that the ball of radius $t/4$ around $c_1$ does not contain any point of $P_2$. What we are able to prove is that the above holds if we exclude a certain small set of coordinates from the computation.
Namely, if $I_{1,2}$ is the set of indices obtained from $[d]$ by removing the $\Delta$ indices that maximize $|(c_1)_r - (c_2)_r|$ over $r \in [d]$, then for $t' = d^{I_{1,2}}(c_1, c_2)$ (\emph{i.e.}, the distance between $c_1$ and $c_2$ when considering only coordinates in $I_{1,2}$), no point of $P_2$ is at distance less than $t'/4$ from $c_1$. Moreover, we can also show that either the ball of radius $t'/4$ around $c_1$ removes enough points from the sampling pool, or the coordinates of $I_{1,2}$ are ``useless'', meaning that we can set the two centers to be equal on these coordinate and still obtain a good approximation.  
The main difficulty here is that we do not know the coordinates $I_{1,2}$ and guessing this set would add a factor $\binom{d}{\Delta}$ to the running time, which we can not afford. We will show how to deal with this problem in Section~\ref{sec:cleaning_sample_pool}, which is the main technical part of our proof. 

\section{Preliminaries}
Throught the paper, we use $[n]$ to denote the set $\{1, \ldots, n\}$ for any integer $n$, and $\mathbb{R}^+$ to denote the set of positive real numbers. For two sets $H$ and $F$ of elements of a universe $\mathcal{U}$, we write $F-H$ for $ F \setminus H$.

\subsection{Points with Missing Coordinates}
 As explained, the goal of the paper is to study clustering of points in $\mathbb{R}^d$ with missing entries in some coordinates. Let us denote the missing entry value by \?
and let $\mathbb{H}^d$ denote the set of elements of $\mathbb{R}^d$ where we allow some coordinates to take the value \?. We say that a point $x=( x_1, \dots, x_d)$ in $\mathbb{H}^d$ is a $\Delta$-point, if at most $\Delta$ of the coordinates  $x_i$ of $x$ have value \?. 

\begin{definition}[\textbf{Domain of a  point}]
 For an element $x \in \mathbb{H}^d$, we call the \emph{domain} of $x$, denoted by $\dom(x)$, the set of coordinates $i \in [d]$ such that $x_i \ne\  \?$. 
\end{definition}

\begin{definition}[\textbf{\fd\, and \pd\, points}]
For a set $S$ of elements of $\mathbb{H}^d$ and a set $I$ of indices in $[d]$, corresponding to coordinates, let $\fd(S,I)$ denote the set of points in $S$ that are \emph{fully defined on $I$}, \emph{i.e.} $x \in S$ such that $\dom(x) \subseteq I$. Formally, \[\fd(S,I)= \{x\in S\, \mid\,  \dom(x) \subseteq I\}.\] By $\pd(S,I)$, we denote the set of points in $S$ that are \emph{partially defined on $I$}, \emph{i.e.} $x \in S$ such that $\dom(x) \cap I \not = \emptyset$. Formally, \[\pd(S,I)= \{x\in S\, \mid\,  \dom(x) \cap I\neq \emptyset\}.\] 
\end{definition}

 With a slight abuse of notation, in all the definitions here and next that concern a particular set of indices $I \subset [d]$, we might use $i \in [d]$ instead of $\{i\}$, \emph{e.g.}, $\pd(S, i) = \pd(S, \{i\})$.

For elements $x,y\in \mathbb{H}^d$ and a set of coordinates in $I\subseteq [d]$, we define \[d^I(x,y) = \sqrt{ \sum_{i \in I} (x_i - y_i)^2 },\] where by convention if either $y = \?$ or $x = \?$, then $(x- y)^2=0$. When $I = [d]$, we let $d(x, y) =d^I(x,y)$. Note that $d(x, y)$  corresponds to the standard Euclidean distance when $x$ and $y$ are elements of $\mathbb{R}^d$.  
For a set $P$ of elements of $\mathbb{H}^d$ and a set $I$ of coordinates, we denote by $c^I(P)$
the ``mean'' of $P$ on the coordinates of $I$. That is, $c^I(P)$ is the element of $\mathbb{H}^d$ such that {for every } $i \in I$,
\begin{equation*}
 c^I(P)_i =
\begin{cases}
\text{\?} & \text{ if $\pd(P, i)$ is empty,}\\
\frac{\sum_{x \in \pd(P, i)}x_i}{|\pd(P, i)|} & \text{otherwise}.\\
\end{cases}       
\end{equation*}

When $I$ contains all elements of $[d]$, we let $c(P) = c^I(P)$. For an element $y \in \mathbb{H}^d$, a set $X$ of elements of $\mathbb{H}^d$ and a set $I$ of coordinates in $[d]$, let us define \[f^{I}_2(X,y) = \sum_{x \in X}d^I(x,y)^2.\] Note that if $I_1$ and $I_2$ are disjoint sets of coordinates, then $f^{I_2  \cup I_1}_2(X,y) = f^{I_1}_2(X,y) + f^{I_2}_2(X,y)$. For $I=[d]$, we write $f_2(X,y)=f^{I}_2(X,y)$.

 For a set $(P_i)_{i \in [k]}$ of subsets of $\mathbb{H}^d$ and a set of points $(c_i)_{i \in [k]}$, we set 
 \[\val\left((P_i)_{i \in [k]},(c_i)_{i \in [k]} \right) = \sum_{i \in [k]}f_2(P_i, c_i).\]

\begin{lemma}\label{lem:approx_center_1}
    For every point $x \in \mathbb{H}^d$, set of points $P\subseteq \mathbb{H}^d$, and set of coordinates $I\subseteq [d]$, it holds that
    \[ f^{I}_2(P,x) = f^{I}_2(P,c(P)) + \sum_{i \in I} |\pd(P,i)| (x_i - c(P)_i)^2. \]

    In particular,
    \[ f^I_2(P,x) \leq f^I_2(P,c(P)) + |\pd(P,I)| d(x, c(P))^2.\]
\end{lemma}

\begin{proof}
For an index $i\in I$, we have 
    \begin{align*}
        f^i_2(P,x) &= \sum_{v \in \pd(P,i)} (v_i - x_i)^2\\
                &= \sum_{v \in \pd(P,i)}(v_i - c(P)_i + c(P)_i - x_i)^2\\
                &= \sum_{v \in \pd(P,i)} (v_i - c(P)_i)^2 + \sum_{v \in \pd(P,i)} (x_i - c(P)_i)^2,
    \end{align*}
    because $\sum_{v \in \pd(P,i)} (v_i - c(P)_i) = 0$ by definition of $c(P)$. This means that $f^i_2(P,x) = f^i_2(P,c^i(P)) + |\pd(P,i)|(x_i - c(P)_i)^2$. We conclude by summing over all $i \in I$.  
\end{proof}

\subsection{$k$-means for $\Delta$-points}
Let us define the {\sc $k$-means clustering} problem for $\Delta$-points. Given an instance $P$ of $n$ $\Delta$-points in $\mathbb{H}^d$, the task  is to partition $P$ into $k$ sets $(P_1, \dots, P_k)$, which we will refer to as \textit{clusters}. A solution also consists of a set of centers $(c_1, \dots, c_k)$ and the objective is to minimize $\sum_{i \in [k]}f_2(P_i, c_i)$. Note that, by Lemma \ref{lem:approx_center_1}, for a given cluster $P_i$ the optimal center is exactly $c(P_i)$, and we can equivalently minimize the objective value $\sum_{i \in [k]}f_2(P_i, c(P_i))$ over all partitions $(P_1, \ldots, P_k)$. Furthermore, notice that if $\dom(x)=\emptyset$, then $x$ always contributes zero to $\sum_{i \in [k]}f_2(P_i, c_i)$, so we can assume that $\dom(x)\neq \emptyset$ for all $x\in P$, and, consequently, $\Delta<d$.

From now on we fix an instance of the {\sc $k$-means clustering} problem, and denote by $P$ the corresponding set of $\Delta$-points in $\mathbb{H}^d$.

\paragraph{Partial clustering.}
Suppose $(P_1, \dots, P_k)$ together with centers $(c_1, \dots, c_k)$ is an optimal solution. As explained earlier, the goal of our algorithm is to discover the centers of the clusters step by step, while assigning some elements of $P$ to some clusters. For this purpose we define the notion of  \textit{partial clustering}. 
We say that  integers $(n_1, \dots, n_k)$, sets $(H_1, \dots, H_k)$, $(I_1, \dots,  I_k)$ and points $(u_1, \dots, u_k)$ form a \textit{partial clustering} if for every $i \in [k]$:   
\begin{itemize}
    \item $H_i$ is a set of elements of $P$,
    \item $n_i$ is an integer in $[\Delta +1]$,
    \item $I_i$ is a set of indices in $[d]$, 
    \item if $n_i > 0$, then $|I_i| \geq d - \Delta + (n_i - 1)$, 
    \item if $n_i = 0$, then $I_i$ and $H_i$ are empty, and
    \item $u_i$ is a point of $\mathbb{H}^d$ such that $\dom(u_i) = I_i$.
\end{itemize}  

Intuitively, for every $i \in [k]$, $H_i$ is a set of points which are already assigned to the cluster $i$, $u_i$ is a partially discovered center of the cluster, and $I_i$ represents the coordinates where $u_i$ is already specified. As we will incur some error each time we are performing a sampling step, the values $n_i$ represent the number of sampling steps that has been done for guessing $u_i$ on $I_i$, and the fact that $I_i \geq d - \Delta + (n_i - 1)$ is used to show that the number of sampling steps performed before reaching a point where $I_i = [d]$ for each $i$ is just $\Oh(\Delta \cdot k)$. Let $R$ denote the set $P - (\bigcup_{i \in [k]} H_i)$ of the points in $P$ that are not yet assigned to a cluster. 

For a partial clustering $\mathcal{P} = \{ (n_i)_{i \in [k]}, (H_i)_{i \in [k]}, (I_i)_{i \in [k]}, (u_i)_{i \in [k]} \}$, an \textit{extension} is a partition $(P'_i)_{i \in [k]}$ of the elements of $P$ such that for every $i \in [k]$, $H_i \subseteq P'_i$. We say that the points $(c'_i)_{i \in [k]}$ are the centers \textit{associated} with $(P'_i)_{i \in [k]}$

if for every $i \in [k]$, $c'_i$ and $u_i$ are equal on the coordinates of $I_i$ and $c'_i$ is equal to $c^{[d] - I_i}(P'_i)$ on $[d] - I_i$.  
The \textit{value} of an extension $(P'_i)_{i \in [k]}$ with associated centers $(c'_i)_{i \in [k]}$, denoted $\val\left((P'_i)_{i \in [k]}\right)$, is equal to $\val\left((P'_i)_{i \in [k]},(c'_i)_{i \in [k]} \right)$. The \textit{value} of a partial clustering $\mathcal{P}$, denoted $OPT(\mathcal{P})$, is the minimum value of an extension $(P'_i)_{i \in [k]}$ of $\mathcal{P}$. We call the extension optimizing this value \textit{optimal}.

\begin{observation}\label{obs:1}
    Let $\mathcal{P} = \{ (n_i)_{i \in [k]}, (H_i)_{i \in [k]}, (I_i)_{i \in [k]}, (u_i)_{i \in [k]}  \}$ be a partial clustering and $x \in R$ such that $\dom(x) \subseteq I_i$ for all $i \in [k]$ and $f_2(x, u_j)$ is minimal among all the $f_2(x, u_i)$. The partial clustering obtained from $\mathcal{P}$ by putting $x$ in the set $H_j$ has the same value as $\mathcal{P}$.
\end{observation}

\begin{proof}
    It follows from the fact that, for any extension with associated centers $(c'_i)_{i\in [k]}$, $f_2(x,c'_i) = f_2(x,u_i)$ for every $i \in [k]$.
\end{proof}

Therefore from now on, we can assume that no point of $ x \in R$ is such that $\dom(x) \subseteq I_i$ for all $i \in [k]$. The previous statement and the conditions on $|I_i|$ imply the following remark: 

\begin{observation}\label{obs:size_ni}
   If $\mathcal{P} = \{ (n_i)_{i \in [k]}, (H_i)_{i \in [k]}, (I_i)_{i \in [k]}, (u_i)_{i \in [k]}  \}$ is a partial clustering such that $\sum_{i \in [k]} n_i  = k(\Delta +1)$, then $\bigcup_{i \in [k]} H_i = P$.  
\end{observation}

Let $\mathcal{P} = \{ (n_i)_{i \in [k]}, (H_i)_{i \in [k]}, (I_i)_{i \in [k]}, (u_i)_{i \in [k]}  \}$ be a partial clustering, and let $(P'_i)_{i \in [k]}$ with centers $(c'_i)_{i \in [k]}$ be its optimal extension. Let us denote $R_i = R \cap P'_i$. The goal of the algorithm is to sample some of the elements of $R$ in order to guess, for some $i \in [k]$, coordinates of $c'_i$ which are not in $I_i$. To do so we need to make sure that our sampling avoids elements $x$ of $R_i$ such that $\dom(x) \subseteq I_i$ ($x\in \fd(R_i, I_i)$) as these elements do not provide any information about $[d] - I_i$. The goal of the following section will be to cluster some of the elements of $\fd(R_i, I_i)$ in order to make this set small compared to $R$. Note that we might need to consider an extension which is not an optimal one.

\section{Finding a Proper Partial Clustering}\label{sec:cleaning_sample_pool}

This section is devoted to the proof of the following lemma:

\begin{lemma}\label{lem:partial_assignment}
    Let $\mathcal{P} = \{ (n_i)_{i \in [k]}, (H_i)_{i \in [k]}, (I_i)_{i \in [k]}, (u_i)_{i \in [k]}  \}$ be a partial clustering. For every constant $\alpha\in \mathbb{R}^+$, $0<\alpha<1$, there exists an algorithm running in time $\Oh(ndk)$, that with probability at least $(\frac{1}{3k^2 \Delta \log(n)^k})$ either: 
    \begin{itemize}
        \item Returns a partial clustering $\mathcal{P}' = \{ (n'_i)_{i \in [k]}, (H'_i)_{i \in [k]}, (I'_i)_{i \in [k]}, (u'_i)_{i \in [k]}  \}$ with $OPT(\mathcal{P}') \leq (1 + \alpha)OPT(\mathcal{P})$ and $\sum_{i \in [k]}n'_i > \sum_{i \in [k]}n_i$; or
        \item Finds a set $B$ of elements of $R$ such that there exists an extension $(P'_i)_{i \in [k]}$ of $\mathcal{P}$ with value at most $(1+\alpha)OPT(\mathcal{P})$ and such that $B \subseteq \bigcup_{i \in [k]}\fd(P'_i, I_i)$ and for every $i \in [k]$, there is an index $j \in [k]$ such that $ |\pd(P'_j \cap R,I_i - I_j)| \geq \frac{\alpha}{32\cdot6^{(\Delta+1)k}} |\fd(P'_i \cap R, I_i) - B|. $
    \end{itemize} 
\end{lemma}
The Lemma~\ref{lem:partial_assignment} will serve as the base for one step of our algorithm in Section~\ref{sec:algo}. The basic idea behind our main algorithm is to iteratively extend the partial clustering, until we get a clustering of $P$. In each step it computes a partial clustering $\mathcal{P}' = \{ (n'_i)_{i \in [k]}, (H'_i)_{i \in [k]}, (I'_i)_{i \in [k]}, (u'_i)_{i \in [k]}  \}$ with $OPT(\mathcal{P}') \leq (1 + \alpha)OPT(\mathcal{P})$ and $\sum_{i \in [k]}n'_i > \sum_{i \in [k]}n_i$ with high enough probability. Observations~\ref{obs:1}~and~\ref{obs:size_ni} then allow us to conclude in at most $k(\Delta+1)$ steps. Note that the first case of the Lemma~\ref{lem:partial_assignment} is precisely what we want. On the other hand, the second case, together with the fact that there are at most $k$ clusters and the pigeonhole principle, will allow us to show that there is an index $r\in[k]$ such that $|(P'_r\cap R) -\fd(P'_r \cap R, I_r)| \geq h(k,\Delta,\alpha)|R-B|$ for some function $h$ depending only on $k$, $\Delta$, and $\alpha$ and $(P'_r\cap R) -\fd(P'_r \cap R, I_r)\subseteq R-B$. Hence we can with sufficiently high probability, by sampling in $R- B$, obtain some points from $(P'_r\cap R) -\fd(P'_r \cap R, I_r)$ and a good approximation of the center $c_r$ on some coordinate outside of $I_r$. 

\subsection{Overview of the Proof}

Before presenting our quite technical proof in the next subsection, let us first explain some of the ideas and difficulties encountered. 
Let $\mathcal{P} = \{ (n_i)_{i \in [k]}, (H_i)_{i \in [k]}, (I_i)_{i \in [k]}, (u_i)_{i \in [k]}  \}$ be a partial clustering and $(P_i)_{i \in [k]}$ be an optimal extension with the associated centers $(c'_i)_{i \in [k]}$. Let us denote $R_i = P_i \cap R$ for every $i \in [k]$. 

Suppose first that $\Delta = 0$ and let $(c_i)_{i \in [k]}$ be the centers of an optimal extension. Kumar et al. \cite{KumarSS10} proved a similar statement as Lemma \ref{lem:partial_assignment}, where the first condition is replaced by the statement that two centers can be equal. To do that, assuming that they have a good approximation $u_i$ of one center $c_i$, they consider the ball $B$ with center $u_i$ and radius $t/4$ where $t$ is the minimum over all $j \in [k]$, $j \not = i$ of the distance of $c_i$ to the other centers $c_j$. Because $u_i$ is a good approximation, $B$ contains only elements of $P_i$. Moreover, they are able to show that either $|P_j|$ is large enough compared to $|P_i-B|$, or putting all the points of $P_j$ in $P_i$ gives a good approximation. Unfortunately, this property is not true anymore for $\Delta$-points with $\Delta >0$. 

First note that in the case when $\Delta >0$, as opposed to Kumar et al. \cite{KumarSS10} where all approximate centers $u_i$ are either in $\mathbb{R}^d$ or not set at all, we have approximate centers that are partially set. Now, if we have some two centers $u_i$ and $u_j$ and we want to extend $u_j$ to some coordinates in $\dom(u_i)-\dom(u_j)=I_i-I_j$ then even
if $t = d^{I_i - I_j}(u_i, c'_j)$, it might not be true that the ball with center $u_i$ and diameter $t/4$ contains no elements of $P_j$ which makes the previous argument considerably more difficult to make. To overcome this problem we will consider the coordinates $r$ in $I_i - I_j$ where $d^r(u_i, c'_j)$ is large, separately one by one. 

Let us now fix an index $i$ for the rest of this subsection. For every $j \in [k]$, let $J_j = I_i - I_j$, let $i^j_1, \dots, i^j_{|J_j|}$ be the coordinates in $J_j$, and let $d^j_r = d^{i^j_r}(u_i, c'_j)$ for all $r \in [|J_j|]$. Without loss of generality, we can assume that $d^j_1 \geq d^j_2 \geq \dots \geq  d^j_{|J_j|}$.  
We distinguish two cases depending on whether $|J_j| \geq \Delta +1$ or not.

%Suppose first that $ |J_j| \geq \Delta +1$.
\paragraph*{Case 1:  $ |J_j| \geq \Delta +1$.}
Let us denote by $J'_j$ the set $\{i^j_{\Delta + 1}, \dots, i^j_{|J_j|}\}$ and let $d^j = d^{J'_j}(u_i, c'_j)$. Note that in this case it follows from the definition of partial clustering that $I_j = \emptyset$, $n_j = 0$, and, consequently,  $J_j=I_i$. 

\begin{lemma}\label{lem:balls_distance}
    For every $j\in [k]\setminus \{i\}$ such that $I_j = \emptyset$ and every $x \in R_j$ such that $\dom(x) \subseteq I_i$, it holds that $d(x,u_i) > d^j/4$. 
\end{lemma}
\begin{proof}
	First observe that if $d^j=0$, then the lemma trivially holds and we can assume for the rest of the proof that $d^j>0$.
    %Indeed, let $Q = \dom(x) \cap J'_j$ and 
    Now, note that $|I_i- \dom(x)| \le \Delta$ and $|I_i-J'_j|=\Delta$. Because both sets ($\dom(x)$ and $J'_j$) are subsets of $I_i$, we have $|\dom(x)-J'_j|\ge |J'_j-\dom(x)|$. 
    Moreover, by the definition of $J'_j$ and because $\dom(x)\subseteq I_i$ we have that $\dom(x)- J'_j \subseteq \{i^j_1, \dots, i^j_{\Delta}\}$. 
    Since $d^j_1, \dots, d^j_{\Delta}$
    are larger than any $d^j_{r} $ for $r > \Delta$, we have that $d^{\dom(x)}(u_i, c'_j) \geq d^j$. 
    
    For the sake of contradiction, let us assume for the remainder of the proof that $d(x,u_i) \leq \frac{d^j}{4}$. 
    Since $\dom(x) \subseteq I_i$, we have $d(x,c'_i) = d(x, u_i)$ and $d(x, c'_i)\le \frac{d^j}{4}$. Moreover, because $(P_i)_{i \in [k]}$ is optimal, we have that $d(x, c'_j) \leq d(x, c'_i)\le \frac{d^j}{4}$. Finally, since $x$, $u_i$, and $c'_j$ are points without any \? on $\dom(x)$, the triangle inequality implies that $d^{\dom(x)}(u_i, c'_j) \leq d^{\dom(x)}(u_i, x) + d^{\dom(x)}(x, c'_j) \le 2\frac{d^j}{4}$, which contradicts $d^{\dom(x)}(u_i, c'_j) \geq d^j$.
\end{proof}

Given the above, we show in the following lemma that the set of elements in $\fd(R, I_i)$ at distance at most $d^j/4$ from $u_i$ is basically what is sufficient to include in the set $B$ of elements in $R$ for the index $i$ to satisfy the second case of Lemma~\ref{lem:partial_assignment} with a caveat that if it is not the case, then we can actually set $u_j$ to be the same as $u_i$ on the coordinates of $J'_j$ without introducing too large error. 

\begin{lemma}\label{lem:distance1}
    For every constant $c>0$ and every index $j\in [k]\setminus\{i\}$ such that $I_i - I_j > \Delta$, if $B$ denotes the set of all elements in $\fd(R,I_i)$ at distance at most $d^j/4$ from $u_i$, then:  
    \begin{itemize}
        \item Either $|\pd(R_j,J'_j) - B| \geq c |P_i - B|$; 
        \item or $f^{J'_j}_2(P_j, u_i) - f^{J'_j}_2(P_j, c'_j) \leq 16c f_2(P_i, u_i)$.
    \end{itemize}
\end{lemma}

\begin{proof}
	Recall that $I_i - I_j > \Delta$ implies that $I_j=\emptyset$ and by the definition of partial clustering we have $n_j=0$, $H_j=\emptyset$ and consequently $P_j=R_j$. 
    Suppose that $|\pd(R_j,J'_j) - B| < c |P_i - B|$ (otherwise we are done). We know that $f_2(P_i, u_i) \geq |P_i - B| (d^j/4)^2$ as all the points of $P_i - B$ are at distance at least $(d^j/4)$ from $u_i$. By Lemma \ref{lem:balls_distance}, we have that \[|\pd(R_j,J'_j)| = | \pd(R_j, J'_j) - B| \leq c|P_i - B|.\] However, by Lemma \ref{lem:approx_center_1}, \[f_2^{J'_j}(R_j, u_i) - f_2^{J'_j}(R_j, c'_j) \leq |\pd(R_j, J'_j)|(d^j)^2,\] which implies that  \[f_2^{J'_j}(P_j, u_i) - f_2^{J'_j}(P_j, c'_j) = f_2^{J'_j}(R_j, u_i) - f_2^{J'_j}(R_j, c'_j) \leq 16c f_2(P_i, u_i).\]
\end{proof}

Now let $j$ be the index that minimizes $d^j$
among all indices $j'$ in $[k]\setminus \{i\}$ for which $I_{j'}=\emptyset$ 
(i.e., an index $j$ such that $d^j=\min_{j'\in[k]\setminus\{i\},  I_{j'}=\emptyset} \{d^{j'} \}$). Then the set $B = \{x\in \fd(R,I_i)\,\mid\, d(x, u_i)\le d^j/4 \}$, \emph{i.e.}, the set of all the elements in $\fd(R,I_i)$ at distance at most $d^j/4$ from $u_i$, does not contain any element in $\bigcup_{j'\in [k]\setminus i, I_{j'}=\emptyset}R_{j'}$ (that is, any element of $R_{j'}$ for every $j'\in [k]\setminus\{i\}$ with $I_{j'}=\emptyset$). 
Furthermore, we have that 
\begin{itemize}
\item  either $|\pd(R_j,J_j) - B| \geq c |R_i - B|$, 
\item  or $f^{J'_j}_2(R_j, u_1) - f^{J'_j}_2(R_j, c'_2) \leq 4 \alpha \Delta f_2(R_1, u_1)$. 
\end{itemize}
When the first inequality occurs, this  is the good case. Basically it means that $\pd(R_j,J'_j)$ is large enough so that sampling in $R - B$ avoids $\fd(R_i,I_i)$ with constant probability. Note that even though we do not know $d^j$, we can get a superset of $R - B$ of size at most $2|R-B|$ with probability $\frac{1}{\log n}$ by taking $\frac{n}{2^r}$ furthest points from $u_i$ for some $r\in [\lfloor\log n\rfloor]$.
To deal with the  case when the second inequality holds, let us first show the following lemma, which will be useful throughout the paper.

\begin{lemma}\label{lem:new_partial1}
    Let $\mathcal{P} = \{ (n_i)_{i \in [k]}, (H_i)_{i \in [k]}, (I_i)_{i \in [k]}, (u_i)_{i \in [k]}  \}$ be a partial clustering, $(P_i)_{i \in [k]}$ be an optimal extension of $\mathcal{P}$ with associated centers $(c'_i)_{i \in [k]}$, and $C\in \mathbb{R}^+$. If there exist two indices $i,j \in [k]$ and $I' \subseteq I_i$ such that $I_j$ is empty, $|I'| \geq d - \Delta$, and $f^{I'}_2(P_j, u_i) - f^{I'}_2(P_j, c'_j) \leq C f_2(P_i, u_i)$, then the partial clustering $\mathcal{P'}$ obtained from $\mathcal{P}$ by setting $I_j = I'$, $u_j$ the element of $\mathbb{H}^d$ equal to $u_i$ on the coordinates of $I'$ and \? on the rest and $n_j := 1$ is a partial clustering of value at most $(1+ C)OPT(\mathcal{P})$.  
\end{lemma}

\begin{proof}
    Indeed, $(P_i)_{i \in [k]}$ is an extension of $\mathcal{P}'$. Let $C_j$ be the point of $\mathbb{H}^d$
    such that 
    \begin{equation}
    (C_j)_r =
    \begin{cases}
    (u_i)_r & \text{if $r\in I'$},\\
    (c'_j)_r & \text{otherwise},\\
    \end{cases}       
    \end{equation}
%     obtained from $c'_j$ by replacing  $(c'_j)_r = (u_i)_r$
    %we denote by $C_j$ the point of $\mathbb{H}^d$ obtained from $c'_j$ by setting $(c'_j)_r = (u_i)_r$ on all the indices $r \in I'$ and 
    and let $C_s= c'_s$ for $s \in [k]\setminus \{j\}$. Then 
    \begin{align*}
        \val\left((P_i)_{i \in [k]}, (C_i)_{i \in [k]} \right) &\leq \val\left((P_i)_{i \in [k]}, (c'_i)_{i \in [k]} \right) + C f_2(P_i, u_i)\\
        &\leq (1 + C)OPT(\mathcal{P}). \tag*{\qedhere}
    \end{align*}
\end{proof}
This means that the extension $\mathcal{P}'$ obtained by applying the previous Lemma with $I' = J'_j$ satisfies the first property of Lemma \ref{lem:partial_assignment}. A major problem is that we do not know $J'_j$ and in the worst case there are $d^{\Delta}$ possibilities for $J'_j$, so guessing a feasible set $I'$ for Lemma~\ref{lem:new_partial1} is not  an option here.
We postpone dealing with this problem for later and switch our focus to the second case.

\paragraph{Case 2: $|J_j| \leq \Delta$.}
Let $r=|J_j|$.
Recall that $i^j_1, \dots, i^j_{r}$ denote the coordinates of $J_j=I_i - I_j$ and $r$ is such that $d^j_r$ is minimal. Let us denote by $B_j$ the ball of elements of $\fd(R,I_i)$ at distance at most $d^j_r$ from $u_i$. 
We show now that we can adapt Lemmas~\ref{lem:balls_distance} and \ref{lem:distance1}  
to this setting as well.

\begin{lemma}\label{lem:balls_distance2}
    Suppose $j \in [k]$ is such that $I_i - I_j \leq \Delta$, if $x \in \pd(R_j, I_i - I_j)$, then $d(x,u_i) \geq d^j_r/4$. 
\end{lemma}

\begin{proof}
    Let $i'$ be the index of $I_i - I_j$ such that $i' \in \dom(x)$. By the choice of $i^j_r$, we have that $d^{i'}(u_i, c'_j)> d^j_r$,  and the triangle inequality allows us to conclude. 
\end{proof}

Note that this time the set of elements of $\fd(R,I_i)$ at distance at most $d^j_r/4$ from $u_i$ can contain some elements of $R_j$, but only those such that $\dom(x) \subseteq I_j$, which are the ``useless'' ones for our sampling as they do not give any information for the coordinates outside of $I_j$. An analogous proof to the one of Lemma~\ref{lem:distance1} gives the following result: 

\begin{lemma}\label{lem:distance2}
    For every constant $c\in \mathbb{R}^+$ and $j \in [k]$ such that $r = |I_i - I_j| \leq \Delta$, if we denote by $B$ the set of elements of $\fd(R,I_i)$ at distance at most $d^j_r/4$ from $u_1$, then: 
    \begin{itemize}
        \item Either $|\pd(R_j, i^j_r) - B| \geq c |R_i - B|$; 
        \item or $f^{i^j_r}_2(R_j, u_1) - f^{i^j_r}_2(R_j, c'_2) \leq 16c f_2(P_1, u_1)$.
    \end{itemize}
\end{lemma}

Again the first case is the good one, as $B$ then provides a set such that drawing a sample from $R- B$ has constant probability to avoid $\fd(R_i,I_i)$. In the second case, however, it means that setting $I_j = I_i \cup i^j_r$ and $(u_j)_r := (u_i)_r$ gives a partial clustering of value at most $(1+ 16c)OPT(\mathcal{P})$. In that case, we can guess the index $i^j_r$ uniformly among $I_i - I_j$ and succeed with probability at least $1/\Delta$. Since we do this only when $I_j$ is non-empty, and thus of size at least $d - \Delta$, the number of times we can do this for each $j \in [k]$ is at most $\Delta$. In total, it means that we will perform this guessing only $k\Delta$ times, and it will only contribute $(\Delta)^{k\Delta}$ to the time complexity. 

Therefore, the main problem we have to overcome is the case $|I_j| = 0$ and $f^{J'_j}_2(P_j, u_i) - f^{J'_j}_2(P_j, c'_j) \leq 16c f_2(P_i, u_i)$ where we do not know $J'_j$. The main idea is the following. Recall that $j$ is the index of $[k]\setminus \{i\}$ such that $I_j$ is empty and $d^{j}$ is minimal among all such indices. Suppose $f^{J'_j}_2(P_j, u_i) - f^{J'_j}_2(P_j, c'_j) \leq 16c f_2(P_i, u_i)$ and now consider the next smallest distance $d^{j'}$ and $B'$ the set of elements of $\fd(R,I_i)$ at distance at most $d^{j'}/4$ from $u_i$. If $|\pd(R_{j'},J'_{j'}) - B'| \geq c |P_i - B'|$, then we almost have the set that we want, except that some elements of $\fd(P_j, I_i)$ can belong to $B'$. The idea will be to move these elements to $P_i$ so that $B'$ satisfies the desired properties for this new extension. While we are able to control the value of the new extension, as it increases by at most $16c f_2(P_i, u_i)$, the fact that it stops being an optimal one creates some problems. Moreover, since we want to do this for each $i \in [k]$, we have to also impose some control over the centers associated with these extensions. This is the goal of the next subsection. 
\subsection{Extensions and Useless Sets of Coordinates}

Let us fix some partial clustering $\mathcal{P} = \{ (n_i)_{i \in [k]}, (H_i)_{i \in [k]}, (I_i)_{i \in [k]}, (u_i)_{i \in [k]}  \}$. By Observation \ref{obs:1}, we can assume that for every index $i \in [k]$ such that $I_i - I_j$ is empty for every $j \in [k]$, no point $x$ satisfying $\dom(x) \subseteq I_i$ belongs to $R$. 
For the rest of this section, $\mathcal{P}$ will be fixed, and any extension $(P'_i)_{i \in [k]}$ will refer to an extension of $\mathcal{P}$ unless stated otherwise. We say that an extension is \textit{safe} if for every $x\in R$ and every $i,j \in [k]$ such that $\dom(x) \subseteq I_i$ and $\dom(x) \subseteq I_j$, $x \in P'_i$ implies that $d(x,c'_i) \leq d(x,c'_j)$.

As explained before, the hard case is when $|I_i - I_j| > \Delta$, which means that $I_j$ is empty. To deal with this case, let us first introduce the following notion of useless sets of coordinates.  
Let $t\in \mathbb{R}^+$,  $i,j \in [k]$ such that $I_j=\emptyset$ and $I_i\neq \emptyset$, let $(P'_i)_{i \in [k]}$ be an extension of $\mathcal{P}$, and finally set $R'_i = P'_i \cap R$. 
A set of indices $Z_{i,j} \subseteq I_i$ is a \textit{$t$-useless set of coordinates} for $(P'_i)_{i \in [k]}$ if
\begin{itemize}
    \item $Z_{i,j}$ is either empty or of size at least $d - \Delta$, and
    \item $ f^{Z_{i,j}}_2(R'_j, u_i) - f^{Z_{i,j}}_2 (R'_j, c'_j) \leq t \cdot OPT (\mathcal{P})$.
\end{itemize}
To simplify the quantifications over $i,j\in [k]$ where $Z_{i,j}$ appears, we define $t$-useless sets of coordinates to be empty sets when $I_i = \emptyset$ or $I_j \not = \emptyset$.
Intuitively, $Z_{i,j}$ corresponds to a set of coordinates such that setting $u_j$ to be equal to $u_i$ on these coordinates still gives a good partial clustering. 
The whole argument revolves around modifying the extension by ``moving'' some elements of some $R'_i$ into some $R'_j$. However, by doing this we might change the values of $ |f^{Z_{i,j}}_2(R'_j, u_i) - f^{Z_{i,j}}_2 (R'_j, c'_j))|$, especially we might change the centers associated with the extension. The next lemma allows us to have some control of what happens for $t$-useless sets of coordinates when the changes are ``small''.

\begin{lemma}\label{lem:useless_moving}
    Let $(P'_i)_{i \in [k]}$ be an extension with associated centers $(c'_i)_{i \in [k]}$ of value at most $(1+t_1) \cdot OPT(\mathcal{P})$ and $(Z_{i,j})_{i,j \in [k]}$ be $t$-useless sets of coordinates for $(P'_i)_{i\in [k]}$. 
    Let $(P^1_i)_{i\in [k]}$ be another extension of $\mathcal{P}$, and denote by $X$ the set of points of $P$ belonging to different clusters in $(P^1_i)_{i\in [k]}$ and  $(P'_i)_{i\in [k]}$. For every $x \in X$ such that $x \in P'_{r}$ and $x \in P^1_{s}$ for some $r, s \in [k]$, let $F(x) =f_2(x, c'_s) - f_2(x,c'_r)$. If 
    \begin{itemize}
        \item for every $x \in X$ and $r \in [k]$ such that $x \in P^1_r$, $\dom(x) \subseteq I_r$ and 
        \item  $\sum_{x \in X} F(x)\leq t_2 \cdot OPT(\mathcal{P})$,
    \end{itemize}  
    then $(Z_{i,j})_{i,j \in [k]}$ are $(t + t_1 + t_2)$-useless sets of coordinates for $(P^1_i)_{i \in [k]}$.  
\end{lemma}

\begin{proof}
    Let $(c^1_i)_{i\in [k]}$ be the centers associated with the extension $(P^1_i)_{i \in [k]}$ and for each $r \in [k]$ let $R^1_r = R \cap P^1_r$. Note that by the definition of $X$, for every $r \in [k]$, every element $x$ of $P^1_r - P'_r$ is such that $\dom(x) \subseteq I_r$. 
    In particular, if for some $r\in [k]$ it holds that $I_r=\emptyset$, then $P^1_r \subseteq P'_r$ (as we assume that no point $x$ in $P$ is such that $\dom(x)=\emptyset$).

     Moreover, the $t$-useless sets of coordinates $Z_{i,j}$, for some $i,j\in [k]$, are only defined in the case that $I_j=\emptyset$.
    This implies that, for all $i,j \in [k]$ such that $Z_{i,j}$ is defined, it holds that: 

    \begin{equation}\label{eq:Zij_subsets}
    	f^{Z_{i,j}}_2(P^1_j, u_i) \leq f^{Z_{i,j}}_2(P'_j, u_i)\quad \text{and} \quad f^{[d]-Z_{i,j}}_2(P^1_j, u_i) \leq f_2^{[d]-Z_{i,j}}(P'_j, u_i).
    \end{equation}

    Suppose now that for some $i,j\in [k]$ such that $Z_{i,j}$ is defined, it holds that $f^{Z_{i,j}}(P^1_j,c^1_j) < f^{Z_{i,j}}(P'_j,c'_j) - (t_1 + t_2)OPT(\mathcal{P})$. We want to reach a contradiction by showing that in this case the value of  $(P^1_i)_{i \in [k]}$ is smaller than $OPT(\mathcal{P})$. For this purpose, let $C_j$ be the element of $\mathbb{R}^d$ equal to $c^1_j$ on the coordinates of $Z_{i,j}$ and to $c'_j$ on the rest of coordinates, and let $C_r = c'_r$ for $r \in [k]\setminus\{j\}$. By the definitions of $X$, $F(\bullet)$, $f_2(\bullet, \bullet)$, and by~(\ref{eq:Zij_subsets}), we have that 
    %\[ 
    \begin{equation}\label{eq:value_bound}
    \sum _{r \in [k], r \not = j}f_2(P^1_r, C_r) + f^{[d] - Z_{i,j}}_2(P^1_j, C_j) \leq \sum _{r \in [k], r \not = j}f_2(P'_r, C_r) + \sum_{x \in X} F(x) + f^{[d] - Z_{i,j}}_2(P'_j, C_j).
    \end{equation}
 %   \]
    Note that 
    \[\val\left((P^1_r)_{r \in [k]}, (C_r)_{r \in [k]}\right) = \sum _{r \in [k], r \not = j}f_2(P^1_r, C_r) + f^{[d] - Z_{i,j}}_2(P^1_j, C_j) + f^{Z_{i,j}}_2(P^1_j,c^1_j)\]
    and 
    \[ \val\left((P'_r)_{r \in [k]}, (c'_r)_{r \in [k]}\right) = \sum _{r \in [k], r \not = j}f_2(P'_r, C_r) + f^{[d] - Z_{i,j}}_2(P'_j, C_j) + f^{Z_{i,j}}_2(P'_j,c'_j).\]
    Hence the inequality~(\ref{eq:value_bound}) together with our assumption implies that \[\val((P^1_r)_{r \in [k]}, (C_r)_{r \in [k]}) < \val((P'_r)_{r \in [k]}, (c'_r)_{r \in [k]}) + t_2 \cdot OPT(\mathcal{P}) - (t_1 + t_2)OPT(\mathcal{P}) < OPT(\mathcal{P}), \]
    a contradiction. 
    This means that \[\left(f^{Z_{i,j}}_2(P'_j,c'_j) - f^{Z_{i,j}}_2 (P^1_j, c'_j)\right) \leq (t_1 + t_2)OPT(\mathcal{P}),\] and thus
    \begin{align*}
        f^{Z_{i,j}}_2(P^1_j, u_i) - f^{Z_{i,j}}_2 (P^1_j, c'_j) &\leq f^{Z_{i,j}}_2(P'_j, u_i) - f^{Z_{i,j}}_2 (P^1_j, c'_j)\\
        &\leq \left(f^{Z_{i,j}}_2(P'_j, u_i) - f^{Z_{i,j}}_2(P'_j,c'_j)\right) +\left(f^{Z_{i,j}}_2(P'_j,c'_j) - f^{Z_{i,j}}_2 (P^1_j, c'_j)\right)\\
        & \leq (t + t_1 + t_2) \cdot OPT (\mathcal{P}).\tag*{\qedhere}
    \end{align*}

\end{proof}

An important part of the proof of Lemmas \ref{lem:balls_distance} and \ref{lem:balls_distance2} is that the extension we consider is optimal. This allowed us to say that picking balls of small radius around $u_i$ will contain only elements of $P'_i$. Since we have to handle extensions which are not optimal, we have to require similar properties for them. This is the role of the next definition. 

For an extension $(P'_i)_{i \in [k]}$ with associated centers $(c'_j)$ and a $t$-useless set of coordinates $Z_{i,j}$ for some $i,j \in [k$], we say that $Z_{i,j}$ is \textit{compatible} with $(P'_i)_{i\in [k]}$ if 

\begin{itemize}
    \item there is no element $x \in P'_j$ such that $\dom(x) \not \subseteq I_j$ and $\dom(x) \subseteq Z_{i,j}$; and
    \item there is no element $x \in P'_j$ such that $\dom(x) \subseteq I_i$ and $d(x,u_i) < \frac{d(x, c'_j)}{2}$. 
\end{itemize} 

Note that the factor $1/2$ here might seems strange, since if $d(x,u_i) \leq d(x,c'_j)$, putting $x$ in $P'_i$ decreases the value of the extension. However, the problem with the definition without the factor $1/2$ would be the following: suppose we have an extension $(P'_i)_{i \in [k]}$, a set of $t$-useless sets of coordinates $Z_{i,j}$ and the goal is to make $Z_{i, j}$ compatible with  $(P'_i)_{i \in [k]}$. To do so, suppose we move iteratively elements $x \in P'_j$ such that $\dom(x) \subseteq I_i$ and $d(x,u_i) < d(x,c'_j)$  to $P'_i$ (where $c'_j$ is the updated center), until there are no such elements. The value of the extension can only decrease after this modification. However, we have to show that the $(Z_{i,j})$ are still $t'$ useless for some bounded $t'$. As we have seen in Lemma \ref{lem:useless_moving}, if we have some upper bound on the sum of $f_2(x, c'_r) - f_2(x,c'_j)$ over the elements $x$ which are moved, then we achieve our goal. If $d(x,u_i) < d(x,c'_j)$, it doesn't seem like we can have control over these values, however if $d(x,u_i) < \frac{d(x, c'_j)}{2}$ every time we move an element $x$, the value of the extension decrease by at least $f_2(x, c'_r)/2$. This means that if we start with an extension of value $(1 + t_1)OPT(\mathcal{P})$, then the sum of $f_2(x, c'_r)$ over the elements $x$ which are moved is bounded by $2t_1 OPT(\mathcal{P})$ and we can apply Lemma \ref{lem:useless_moving}. This will be the main argument of the next lemma.

\begin{lemma}\label{lem:useless_next}
    Suppose $(P'_i)_{i \in [k]}$ is a safe extension of value at most $(1+t_1) OPT(\mathcal{P})$ with associated centers $(c'_j)$ and $(Z_{i,j})_{i,j \in [k]}$ are compatible $t_2$-useless sets of coordinates for $(P'_i)_{i\in [k]}$. If there exist $i,j \in [k]$ with $I_i \not = \emptyset$ and $I_j= \emptyset$ as well as a set $I_{i,j}$ of coordinates in $I_i$ such that $|I_{i,j}| \geq d - \Delta$, $I_{i,j} \not \subseteq  Z_{i,j}$, and $f^{I_{i,j}}_2(P'_j, u_i) - f^{I_{i,j}}_2(P'_j, c'_j) \leq t_3\cdot OPT(\mathcal{P})$, then there exists a safe extension $(P^1_i)_{i \in [k]}$ of $\mathcal{P}$ with compatible $3t_1 + 2t_3 + t_2$-useless sets of coordinates $(Z^1_{i,j})_{i,j}$ such that the value of $(P^1_i)_{i\in [k]}$ is at most $(1 + t_1 + t_3)OPT(\mathcal{P})$ and $\sum_{i,j \in [k]}|Z^1_{i,j}| >\sum_{i,j \in [k]}|Z'_{i,j}|$.
\end{lemma}

\begin{proof}
    Define $Z^1_{i,j} = Z_{i,j} \cup I_{i,j}$, and note that $|Z^1_{i,j}| > |Z_{i,j}|$ and $|Z^1_{i,j}| \geq |d - \Delta|$. Let $Z^1_{i',j'} = Z_{i',j'}$ for any other pair $i',j' \in [k]$. Note that $Z^1_{i',j'}$ is still a compatible $t$-useless set of coordinates for $(P'_i)_{i \in [k]}$. Because $f^{I_{i,j}}_2(P'_j, u_i) - f^{I_{i,j}}_2(P'_j, c'_j) \leq t_3\cdot OPT(\mathcal{P})$ and $Z_{i,j}$ is $t_2$-useless, we get that $Z^1_{i,j}$ is a $t_3 + t_2$ useless set of coordinates for $(P'_i)_{i\in [k]}$.
    
    Consider now the extension $(P^1_r)_{r \in [k]}$ obtained from $(P'_r)_{r \in [k]}$ by putting in $P'_i$ the set $X_1$ of elements $x \in P'_j$ such that $\dom(x) \not \subseteq I_j$ and $\dom(x) \subseteq Z^1_{i,j}$. Let us define a sequence $(P^2_r)_{r \in [k]}, \ldots (P^q_r)_{r \in [k]}$ of extensions with the associated centers $(c^2_r)_{r \in [k]}$, \ldots, $(c^q_r)_{r \in [k]}$. For each $s \in [q - 1]$, we obtain the extension $(P^{s+1}_r)_{r \in [k]}$ from $(P^s_r)_{r \in [k]}$ as follows. %As long as 
    If there is an element in $P^s_j$ such that $\dom(x) \subseteq I_{i'}$ for some $i' \in [k]$ and $d(x,u_{i'}) \leq \frac{d(x,c^s_j)}{2}$ then we get $(P^{s+1}_r)_{r \in [k]}$ from $(P^s_r)_{r \in [k]}$ by moving $x$ from $P^s_j$ to $P^s_{i'}$. If there are multiple choices, we simply take any $ i'$ such that $d(x,u_{i'})$ is minimal. Note that throughout this process, we are only removing elements from $P^s_j$ to put it in another $P^s_{i'}$, which means that this process has to stop after at most $|P^1_j|$ steps. Let $(P''_r)_{r \in [k]}$ be the final extension of this process.
    % $(P''_i)$. 
    Note as well that since we only add to $ P^s_{i'}$ elements $x$ such that $\dom(x) \in I_{i'}$, it means that $c^s_{i'} = c'_{i'}$ for every $i'$ different from $j$. Denote by $X_2$ the set of all elements we moved between $(P^1_r)_{r \in [k]}$ and $(P''_r)_{r \in [k]}$, that is, $X_2=P^1_j - P''_j$. For $x \in X_2$ such that $x \in P''_r$, let $f(x) = f_2(x,c'_r)$. 

    \begin{claim}
        $\sum_{x \in X_2} f(x) \leq 2(t_1 + t_3)OPT(\mathcal{P})$. 
    \end{claim}

    \begin{cProof}
       Since $f^{I_{i,j}}_2(P'_j, u_i) - f^{I_{i,j}}_2(P'_j, c'_j) \leq t_3\cdot OPT(\mathcal{P})$, we have that $(P^1_r)_{r \in [k]}$ has value at most $(1+t_1 + t_3)OPT(\mathcal{P})$. However, suppose $x$ is such that $x \in P^s_j$ and $x \in P^{s+1}_r$ for some $r \in [k]\setminus j$. It means that $d(x,u_r) \leq \frac{d(x, c(P^s_j))}{2}$ and: \[\val((P^{s+1}_r)_{r \in [k]}) = \val((P^{s}_r)_{r \in [k]}) + (d(x,u_r)^2 - d(x, c(P^s_j))^2) \leq val((P^{s}_r)_{r \in [k]}) - \frac{3f(x)}{4}.\] Since $\val((P''_i)_{i\in [k]}) \geq OPT(\mathcal{P})$, this ends the proof of the claim. 
    \end{cProof}

    Overall, if we set $X := X_1 \cup X_2$, then we can apply Lemma \ref{lem:useless_moving} to show that each $Z^1_{i',j'}$ is a $3t_1 + 2t_3 + t_2$-useless set of coordinates for $i',j' \in [k]$. 
    Moreover, since we only move elements of $P'_j$ and $c'_r = c''_r$ for every other $r \in [k]\setminus j$, the extension $(P''_r)_{r \in [k]}$ remains safe. Finally we can verify that the $(Z^1_{i',j'})$ are compatible with $(P''_r)_{r \in [k]}$ as well. Indeed, for every $i' \in [k]$,  $Z^1_{i',j}$ is compatible with $(P''_r)_{r \in [k]}$ by the definition of $X$. For every $j' \not = j$ and $i' \in [k]$, we have that $Z^1_{i',j'} = Z_{i',j'}$ and thus since $Z_{i',j'}$ is compatible with $(P'_i)_{i \in [k]}$ and the elements $x$ of $P''_{j'} - P'_{j'}$ are such that $\dom(x) \subseteq I_{j'}$, there is no element $x \in P''_{j'}$ such that $\dom(x) \not \subseteq I_{j'}$ and $\dom(x) \subseteq Z^1_{i', j'}$. Moreover, since $c''_{j'} = c'_{j'}$, we have $d(x,c''_{j'}) = d(x,c'_{j'})$, which means that there is no element of $P'_j$ such that $d(x,u_{i'}) \leq \frac{d(x,c''_j)}{2}$. Finally for every $x \in P''_{j'} - P'_{j'}$, $d(x, c''_{j'}) = d(x, u_{j'})$, which ends the proof, as we chose $j'$ as the index such that $\dom(x) \subseteq I_{j'}$ and $ d(x, u_{j'})$ is minimal. 
\end{proof}

The next lemma is the main technical part of this subsection. 

\begin{lemma}\label{lem:partial_assignment_induc}
    Let $t\in \mathbb{R}^+$, let $(P'_i)_{i\in [k]}$ be a safe extension of $\mathcal{P}$ with value at most $(1 + t)OPT(\mathcal{P})$ and for every $i,j \in [k]$ let $Z_{i,j}$ be a compatible $t$-useless set of coordinates, such that $(Z_{i, j})$ are compatible with $(P'_i)_{i \in [k]}$. One of the following holds:
    \begin{itemize}
        \item Either $Z_{i,j} = I_i$ for one pair $i, j \in [k]$ such that $I_i \not = \emptyset$ and $I_j = \emptyset$; or
        \item There exists $i, j\in [k]$ such that $I_i$ and $I_j$ are nonempty and $f^{i^j_1}_2(P'_j, u_i) - f^{i^j_1}_2(P'_j, c'_j) \leq t\cdot f_2(R'_i, u_i)$ for some index $i^j_1 \in I_i - I_j$; or
        \item There is an algorithm running in time $\Oh(nkd)$ that returns for every $i \in [k]$ a set $T'_i$ such that $T'_i \subseteq \fd(R'_i, I_i)$ and $|\pd(P'_j, I_i - I_j) - T'_i| \geq \frac{t}{32} |\fd(R'_i, I_i) - T'_i|$ for some $j \in [k]$ with probability at least $(\frac{1}{\log(n)})^k$; or 
        \item There exists an extension $(P^1_i)_{i \in [k]}$ of $\mathcal{P}$ with compatible $(5t + t^2)$-useless sets of coordinates $(Z^1_{i,j})_{i,j\in [k]}$ such that the value of $(P^1_i)_{i \in [k]})$ is at most $(1 + \frac{3t+t^2}{2}) OPT(\mathcal{P})$\ee{I updated the constants according to change in Lemma 12, somebody might want to check it ...} and $\sum_{i,j \in [k]}|Z^1_{i,j}| > \sum_{i,j \in [k]}|Z_{i,j}|$.
    \end{itemize}
\end{lemma}

\begin{proof}
    Let $c = t/32$.
    If $Z_{i,j} = I_i$ for one pair $i, j \in [k]$ such that $I_j = \emptyset$ and $I_i \not = \emptyset$, then nothing needs to be done, so let us assume this is not the case. Let $F$ be the set of indices $i \in [k]$ such that $I_i$ is empty and $G = [k] - F$.

    Let $i$ be some index of $G$. 
    For any $j \in F$, let us denote $i_{j,min}$ the element of $I_i - Z_{i,j}$ such that $d^{i_{j,min}}(u_i, c(P'_j))$ is minimal. Let us now denote by $r_{1,j}, \dots, r_{\Delta,j}$ the $\Delta$ coordinates of $I_i - i_{j,min}$ such that the $d^{r_{s,j}}(u_i, c(P'_j) )$ for $s\in [\Delta]$ are the $\Delta$ maximal values among all $d^{r}(u_i, c(P'_j) )$ for $r \in I_i - i_{j,min}$, and set $I_{i,j} = I_i - \{ r_1, \dots, r_{\Delta} \}$. Let $d^j = d^{I_{i,j}}(u_i, c(P'_j))$. The following claim is the analogue of Lemma \ref{lem:balls_distance}, but using the fact that $Z_{i,j}$ is compatible with $(P'_r)_{r \in [k]}$ instead of the extension being optimal.

    \begin{claim}\label{clm:balls_distance1}
        For any $x \in P'_j$ such that $\dom(x) \subseteq I_i$, $d(x,u_i) > d^j/4$. 
    \end{claim}

    \begin{cProof}
        Let $x$ be an element of $P'_j$ such that $\dom(x) \subseteq I_i$ and note that $\dom(x) - I_i <\Delta$. Because $Z_{i,j}$ is compatible with $(P'_i)_{i\in [k]}$, we have that $\dom(x) \cap (I_i - Z_{i,j})$ is nonempty, and thus, by the choice of $I_{i,j}$,
        %$i_{j,min}$ and the $r_{s,j}$, 
        we have that $d^{\dom(x)}(u_i, c'_j) \geq d^{I_{i,j}}(u_i, c'_j)$. However, because  $Z_{i,j}$ is compatible with $(P'_i)_{i\in [k]}$ and $\dom(x) \subseteq I_i$, we have that $d^{\dom(x)}(x,c'_j) \leq 2d(x,u_i)$. Therefore, if $d(x,u_i) \leq  d^{I_{i,j}}(u_i, c'_j)/4 \leq d^{\dom(x)}(u_i, c'_j)/4$, the triangle inequality gives us:  \[d^{\dom(x)}(u_i, c'_j) \leq d^{\dom(x)}(u_i, x) + d^{\dom(x)}(x, c'_j) < d^{\dom(x)}(u_i, c'_j), \]
        where the triangle inequality applies because no coordinates of $x, u_i$ and $c'_j$ has value \? on $\dom(x)$. This is a contradiction and thus ends the proof. 
    \end{cProof}
    We can also show the following: 

    \begin{claim}\label{clm:distance1}
        For every constant $c\in \mathbb{R}^+$ and $j \in F$, the set $B_j$ of elements of $\fd(R,I_i)$ at distance at most $d^j/4$ from $u_i$ is such that one of the following properties is satisfied:
       \begin{itemize}
           \item  $|P'_j - B_j| \geq c |\fd(R'_i, I_i) - B_j|$; or
           \item  $f^{I_{i,j}}_2(P'_j, u_i) - f^{I_{i,j}}_2(P'_j, c'_j) \leq 16 c( 1+ t)OPT(\mathcal{P})$.
       \end{itemize}
       \end{claim}
   
       \begin{cProof}
            Indeed, suppose that $|P'_j - B_i| < c |\fd(R'_i, I_i) - B_j|$ . We know that $f_2(P'_i, u_i) \geq |\fd(R'_i, I_i) - B| (d^j/4)^2$ as all the points of $\fd(R'_i, I_i) - B_i$ are at distance at least $d^j/4$ from $u_i$. Because of 
            Claim \ref{clm:balls_distance1}, we have that $|P'_j - B_j| = |P'_j| \leq c|\fd(R'_i, I_i) - B_i|$. Moreover, since $I_j$ is empty, it means that $c'_j = c(P'_j)$ and thus, by Lemma \ref{lem:approx_center_1}, we have that $f_2(P'_j, u_i) - f_2(P'_j, c'_j) \leq |P'_j|d(u_i,c'_j)^2$, which implies that $f_2(P'_j, u_i) - f_2(P'_j, c'_j) \leq 16c \cdot f_2(P'_i,u_i) \leq 16c(1 + t)OPT(\mathcal{P})$.
       \end{cProof}
       If $f^{I_{i,j}}_2(P'_j, u_i) - f^{I_{i,j}}_2(P'_j, c'_j) \leq 16 c( 1+ t)OPT(\mathcal{P})$, then Lemma \ref{lem:useless_next} gives us the existence of an extension and some sets of coordinates satisfying the fourth property of the lemma. Therefore, from now on we assume that for every $j \in F$, $|P'_j - B_j| \geq c |\fd(R'_i, I_i) - B_j|$.

    For any $j \in G$ different from $i$ such that $I_i - I_j \not = \emptyset$, let us define $i^j_1, \dots, i^j_{|I_i - I_j|}$ as the coordinates of $I_i - I_j$, and set $d^j_r = d^{i^j_r}(u_i, c'_j)$ for all $r \in [|I_i - I_j|]$. Without loss of generality, we can assume $d^j_1$ is minimum, and let $d^j = d^j_1$. Using analogous proofs as those of Lemmas \ref{lem:balls_distance2} and \ref{lem:distance2}, we can use the fact that $(P'_i)_{i\in [k]}$ is safe to prove the following two claims: \ee{Do we want to say that the proof is the same at least? Or do we maybe want to define safe extension earlier and just state the lemmas with a safe extension?}

    \begin{claim}\label{clm:balls_distance2}
       For $j \in G$, if $x \in \pd(P'_j, I_i - I_j)$, then $d(x,u_i) \geq d^j/4$. 
    \end{claim}
   
    \begin{claim}\label{clm:distance2}
        For every constant $c\in\mathbb{R}^+$ and $j \in G$, if we denote by $B_j$ the set of elements of $\fd(R,I_i)$ at distance at most $d^j$ from $u_1$, then: 
        \begin{itemize}
            \item Either $|\pd(P'_j, i^j_1) - B_j| \geq c |R'_i - B|$; or
            \item $f^{i^j_1}_2(P'_j, u_i) - f^{i^j_1}_2(P'_j, c'_j) \leq 16c f_2(P'_i, u_i)$.
        \end{itemize}
    \end{claim}
    Note that if there exists $j$ such that $f^{i^j_1}_2(P'_j, u_i) - f^{i^j_1}_2(P'_j, c'_j) \leq 16c f_2(R'_i, u_i)$, then the second property of the lemma is satisfied, so we can assume that $|\pd(P'_j, i^j_1) - B_j| \geq c |R'_i - B_j|$ for all $j \in G$.

    Let us consider now the index $j \in G \cup F$ such that $d^j$ is defined and minimal. Note that if no $d^j$ is defined, it means that $I_i - I_j$ is empty for every $j$ and thus because of Observation \ref{obs:1}, we can assume that $\fd(R'_i, I_i)$ is empty. Let us denote by $B_i$ the set of elements $x$ of $\fd(R, I_i)$ at distance at most $d^j/4$ from $u_i$ and such that, if $\dom(x) \subseteq I_j$ for some $j \in [k]$, then $d(x,u_i) \leq d(x, u_j)$. By combining Claims \ref{clm:balls_distance1} and \ref{clm:balls_distance2}, the choice of $j$ and the assumptions we made on the results of Claims \ref{clm:distance2} and \ref{clm:distance1} we get the following claim:

    \begin{claim}\label{clm:empty_balls}
         $B_i$ contains only elements of $R'_i$ and $|\pd(P'_j, I_1 - I_j) - B_i| \geq c |\fd(R'_i, I_i) - B_i|$.   
    \end{claim} 
     Consider the integer $r$ such $\frac{n}{2^{r}} \leq |\fd(R,I_i) - B_i| \leq \frac{n}{2^{r-1}}$, let $H$ be the set of $\frac{n}{2^{r-1}}$ points of $\fd(R,I_i)$ that are the farthest away from $u_1$ and let $B'_i = \fd(R,I_i) - H$. By the definition of $H$,  $B'_i \subseteq B_i$ and $|R_i - B'| \leq 2|R_i - B|$ which means that $|\pd(P'_j,
     I_i - I_j) - B'_i| \geq \frac{c|R_i - B'|}{2}$. Therefore, if the algorithm selects uniformly at random an integer in $[\log(n)]$, then with probability $1/\log(n)$ this integer is equal to $r$, and the algorithm is then able to find the set $T'_i := B'_i$. Note that  once $r$ is selected, the set $B'$ can be found in $\Oh(nd)$ time. We finish the proof by repeating this for every $i$ such that $d^j$ can be defined. For the other indices $i$, as explained, $\fd(R'_i, I_i)$ is empty and thus $B'_i := \emptyset$ has the required properties. 
\end{proof}

A very important remark here is that in the case where there is an algorithm running in time $\Oh(nkd)$ that returns for every $i \in [k]$ a set $T'_i$ such that $T'_i \subseteq \fd(R'_i, I_i)$ and $|\pd(P'_j, I_i - I_j) - T'_i| \geq \frac{t}{32} |\fd(R'_i, I_i) - T'_i|$ for some $j \in [k]$ with probability at least $(\frac{1}{\log(n)})^k$, the algorithm does not need to know the extension $(P'_i)_{i \in [k]}$, as the only thing that matters are the distances of the elements of $P$ to the point $u_i$ that are given in $\mathcal{P}$.

\subsection{Proof of Lemma~\ref{lem:partial_assignment}}

We are now ready to prove the main result of this section, Lemma \ref{lem:partial_assignment}. 

\begin{proof}[Proof of Lemma \ref{lem:partial_assignment}]
    Let $t_1 = \frac{\alpha}{6^{(\Delta+1)k}}$
    %$t_1 = \frac{\alpha^{1/2^{(\Delta+1)k}}}{6}$ 
    and note that $6^{{(\Delta+1)k}}t_1 = \alpha$.
    %$(14t_1)^{2^{(\Delta+1)k}} \leq \alpha$. 
    Let $(P^1_i)_{i\in [k]}$ be an optimal extension of $\mathcal{P}$ and $Z^1_{i,j} = \emptyset$ for every pair $i, j \in [k]$. Note that $(P^1_i)_{i\in [k]}$ is safe and the $(Z^1_{i,j})$ are compatible $t_1$-useless sets of coordinates. 
    By applying Lemma \ref{lem:partial_assignment_induc} to $(P^1_i)_{i\in [k]}$ and $(Z^1_{i,j})_{i,j \in [k]}$, we have that, denoting $R^1_i = P^1_i \cap R$:
    \begin{itemize}
        \item Either $Z_{i,j} = I_i$ for one pair $i, j \in [k]$ such that $I_i \not = \emptyset$ and $I_j = \emptyset$; or
        \item There exists $i, j\in [k]$ such that $I_i$ and $I_j$ are non empty and $f^{i^j_1}_2(P^1_j, u_i) - f^{i^j_1}_2(P^1_j, c'_j) \leq t_1\cdot f_2(R^1_i, u_i)$ for some index $i^j_1 \in I_i - I_j$; or
        \item There is an algorithm that returns for every $i \in [k]$ a set $T'_i$ such that $T'_i \subseteq \fd(R^1_i, I_i)$ and $|\pd(P^1_j, I_i - I_j) - T'_i| \geq \frac{t_1}{32} |\fd(R^1_i, I_i) - T'_i|$ for some $j \in [k]$ with probability at least $(\frac{1}{\log(n)})^k$; or 
        \item There exists an extension $(P^2_i)_{i \in [k]}$ of $\mathcal{P}$ with some compatible $5t_1+t_1^2$-useless set of coordinates $(Z^2_{i,j})_{i,j}$ such that the value of $(P^2_i)_{i \in [k]}$ is at most $(1 + \frac{3t_1+t_1^2}{2}) OPT(\mathcal{P})\le (1+5t_1+t_1^2)OPT(\mathcal{P})$ and $\sum_{i,j \in [k]}|Z^2_{i,j}| >\sum_{i,j \in [k]}|Z^1_{i,j}|$.
    \end{itemize}

    In the first case, the partial clustering $\mathcal{P}'$ obtained from $\mathcal{P}$ by setting $I_j := I_i$, $u_j := u_i$, and $n_j := n_i$, has value at most \[ OPT(\mathcal{P}) + \left(f^{Z_{i,j}}_2(P'_j, u_i) - f^{i^j_1}_2(P'_j, c'_j) \right) \leq  (1 + t_1) OPT(\mathcal{P}) \] by definition of $t_1$-useless sets of coordinates, and therefore $\mathcal{P}'$ satisfies the first property of Lemma \ref{lem:partial_assignment}.
    In the second case, the partial clustering  $\mathcal{P}'$ obtained from $\mathcal{P}$ by setting $I_j := I_j \cup i^j_1$, $(u_j)_{i^j_1} := (u_i)_{i^j_1}$ and $n_j := n_j +1$ is a partial clustering of value at most \[ OPT(\mathcal{P}) + \left(f^{i^j_1}_2(P'_j, u_i) - f^{i^j_1}_2(P'_j, c'_j)\right)\leq (1 + t_1) OPT(\mathcal{P}),\] and therefore $\mathcal{P}'$ satisfies the first property of Lemma \ref{lem:partial_assignment}.
    In the third case, then the set $B = \bigcup_{i \in [k]} T'_i$ satisfies the second property of Lemma \ref{lem:partial_assignment}.

    In the last case, we can again apply Lemma \ref{lem:partial_assignment_induc} to $(P^2_i)_{i\in [k]}$ and $(Z^2_{i,j})_{i,j \in [k]}$ with $t_2 = 5t_1+t_1^2$. By repeating this process until we arrive to an application of Lemma \ref{lem:partial_assignment_induc} where one of the first three cases is satisfied, we can define a sequence of extensions $(P^s_i)_{i\in [k]}$, $(Z^s_{i,j})_{i,j \in [k]}$ and $t_s = 5t_{s-1}+t_{s-1}^2$ such that the $(Z^s_{i,j})_{i,j \in [k]}$ are compatible $t_s$-useless sets of coordinates for $(P^s_i)_{i\in [k]}$. Note that if $t_{s-1}\le1$, then $t_s\le 6t_{s-1}$ and recall that $\alpha < 1$. Moreover, this process has to stop after at most $(\Delta+1)k$ iterations, because $\sum_{i,j \in [k]}|Z^{s+1}_{i,j}| > \sum_{i,j \in [k]}|Z^{s}_{i,j}|$, and if $Z^r_{i,j}$ increases more than $\Delta$ times, then $|Z^r_{i,j}|= d$, and we are in the first case. 
    Therefore, by the choice of $t_1$, there exists a safe extension $(P'_i)_{i\in [k]}$ with compatible $\alpha$-useless sets $(Z'_{i,j})_{i, j\in [k]}$ such that the application of Lemma \ref{lem:partial_assignment_induc} falls into one of the first three cases, and we can conclude.

    The desired algorithm proceeds as follows. First, it guesses with probability $1/3$ in which of the above cases it falls. In the first case, guessing the pair $i,j$ allows us to conclude with probability $1/k^2$. In the second, guessing $i,j$ and $i^j_1 \in I_i - I_j$ allows us to conclude with probability $\frac{1}{k^2\Delta}$ (remember that $|I_i - I_j| \leq \Delta$). Finally in the last case, the algorithm succeeds if the algorithm of Lemma \ref{lem:partial_assignment_induc} succeeds, so the probability of that is $(\frac{1}{\log(n)})^k$. 

    Overall, the probability of success is at least $\frac{1}{3k^2\Delta\log(n)^k}$, and the algorithm runs in time $\Oh(nkd)$.   
\end{proof}

\section{The Algorithm}\label{sec:algo}

Now that we have Lemma \ref{lem:partial_assignment}, we can describe our algorithm. 
Let us first recall the following lemma, which is a direct consequence of the definition of variance (see Lemma 1 of Inaba et al.~\cite{inaba1994applications}). 

\begin{lemma}\label{lem:variance}
	Let $x_1, \dots, x_m$ be a set of reals with average $\mu$ and $s_1, \dots, s_t$ be a set of elements obtained by $t$ independent and uniform draws among the $x_i$.  The random variable $s = \sum_{i \in t} s_i/t$ is such that $\mathbb{E}(|s-\mu|^2) \leq \frac{\sum_{i \in [m]}|x_i - \mu|^2 }{tm}$.   
\end{lemma}

The main element of our proof is the following lemma providing one step of the algorithm.

\begin{lemma}\label{lem:algorithm_step}
    For every constant $\alpha\in \mathbb{R}^+$, there exists an algorithm 
    that, given a partial clustering $\mathcal{P} = \{ (n_i)_{i \in [k]}, (H_i)_{i \in [k]}, (I_i)_{i \in [k]}, (u_i)_{i \in [k]}  \}$
    outputs in time $\Oh(nkd)$ with probability at least 
    \[
    \min\{\frac{1}{2^{\Oh(\frac{\Delta^3k}{\alpha}\log\frac{1}{\alpha})}(\log n)^k}, \frac{1}{2^{\Oh({\Delta^6k}\log\Delta)}(\log n)^k}\}
    \]
     a partial clustering $\mathcal{P}' = \{ (n'_i)_{i \in [k]}, (H'_i)_{i \in [k]}, (I'_i)_{i \in [k]}, (u'_i)_{i \in [k]}  \}$ such that $\sum_{i \in [k]}n'_i > \sum_{i \in [k]}n_i$ and $OPT(\mathcal{P'}) \leq (1 + \alpha )OPT(\mathcal{P})$.

\end{lemma}

\begin{proof}%[Proof of Lemma \ref{lem:algorithm_step}]
    Let us fix  sufficiently small $q\in \mathbb{R}^+$ such that $(1+q)^2 \leq (1 + \alpha)$ and $\exp(-\frac{1}{4\Delta q})\le \frac{q}{4\Delta}$.
    %\ee{We might want to check that, but I think setting $q< \frac{1}{20\Delta^2}$ suffices}. 
    The choice of $q$ will be clear later in the proof. For now just notice that $q<1$ and setting $q=\min\{\frac{\alpha}{3}, \frac{1}{128\Delta^3}\}$ satisfies both conditions. %if $\alpha$ is sufficiently small (this bound depends only on $\Delta$), we can just let $q = \frac{\alpha}{3}$. 
    Additionally,  let $q' = \frac{q}{32\cdot6^{(\Delta+1)k}}$. 
    By applying Lemma \ref{lem:partial_assignment} with the constant $q$, we have an algorithm that runs in time $\Oh(nd)$ and with probability at least $(\frac{1}{2 \log(n)})^{k}$ returns either: 
    
    \begin{itemize}
        \item A partial clustering $\mathcal{P}' = \{ (n'_i)_{i \in [k]}, (H'_i)_{i \in [k]}, (I'_i)_{i \in [k]}, (u'_i)_{i \in [k]}  \}$ with $OPT(\mathcal{P}') \leq (1 + q)OPT(\mathcal{P})$ and $\sum_{i \in [k]}n'_i > \sum_{i \in [k]}n_i$; or
        \item a set $B$ of elements of $R$ such that there exists an extension $(P'_i)_{i \in [k]}$ of $\mathcal{P}$ with value at most $(1+q)OPT(\mathcal{P})$ and such that $B \subseteq \bigcup_{i \in [k]}\fd(P'_i \cap R, I_i)$ and for every $i \in [k]$, there is an index $j \in [k]$ such that $ |\pd(P'_j \cap R,I_i - I_j) - B| \geq q' |\fd(P'_i \cap R, I_i)  -B|. $
    \end{itemize} 

    In the former case, nothing needs to be done as $\mathcal{P}'$ satisfies all the properties of the lemma. Therefore, from now on we assume that we are in the latter case and we are given a set $B\subseteq R$ satisfying all the conditions of the second case of Lemma~\ref{lem:partial_assignment}. Let $(P'_i)_{i \in [k]}$ be the hypothetical extension with value at most $(1+q)OPT(\mathcal{P})$ whose existence is guaranteed.
    %the existence of an extension $(P'_i)_{i \in [k]}$ of $\mathcal{P}$ with value at most $(1+q)OPT(\mathcal{P})$ and a set $B \subseteq R$ such that $B \subseteq \bigcup_{i \in [k]}\fd(P'_i \cap R, I_i)$ and for every $i \in [k]$, there is an index $j \in [k]$ such that $ |\pd(P'_j \cap R,I'_i - I'_j) - B| \geq q' |\fd(P'_i \cap R, I'_i) -B|$. 
    Let $R'_i  = P'_i \cap (R - B)$ for all $i \in [k]$ and let $(c'_i)_{i \in [k]}$ be the centers associated with $(P'_i)_{i \in [k]}$. Let $i$ be the index such that $|R'_i|$ is maximal, meaning $|R'_i| \geq |R- B|/k$. Now either $|\fd(R'_i, I_i)| \geq |R'_i|/2$, in which case we know that there exists an index $j \in [k]$ such that $|\pd(R'_j,I_i - I_j)| \geq q' |\fd(R'_i, I_i)| \geq \frac{q'|R-B|}{2k}$, or $|R'_i -\fd(R'_i, I_i)| \geq \frac{|R-B|}{2k}$. Note that $\pd(R'_j,I_i - I_j) \cap \fd(R'_j, I_j) = \emptyset$ which means that there exists an index $r \in \{i,j\}$ %(either $i$ or $j$ depending on the cases) 
    such that $|R'_r -\fd(R'_r, I_r)| \geq \frac{q'|R-B|}{2k}$. 
    The goal of the algorithm now will be to sample points inside $R'_r -\fd(R'_r, I_r)$ in order to obtain a good estimate of some coordinates of $c'_r$ that are not yet in $I_r$. We will consider two different cases, depending on whether $|I_r| \geq d - \Delta$ or not. 
    
%\textbf{ Case 1: $|I_r| < d- \Delta$}\\
\paragraph*{Case 1: $|I_r| < d- \Delta$.}
Note that in this case, by the definition of a partial clustering, $n_r = 0$ and $H_r \cup \fd(R'_r, I_r)$ is empty. This implies that $R'_r = P'_r$.
Let $\delta = \frac{1}{2\Delta}$
%$\delta = 1 - (1/2)^{1/\Delta}$ 
and note that $(1 -\delta)^\Delta \geq 1/2$.
% and $\frac{1}{2\Delta}\le \delta\le \frac{1}{\Delta}$\ee{We can maybe directly set $\delta= \frac{1}{2\Delta}$}.
We claim that there exists a set $L_r$ of at most $\Delta$ coordinates (possibly empty) such that there exists a set $F_r \subseteq R'_r$, $|F_r|\geq |R'_r|/2$ where every element $x$ of $F_r$ is such that $x_j =\ \?$ on every $j \in L_r$ and for every $i \in [d] - L_r$, there are at most $(1- \delta)|F_r|$ points $x$ in $F_r$ with $x_i =\ \?$. 
We can obtain the set $F_r$ from $R'_r$ as follows. Start with $L_r = \emptyset$ and $F_r = R'_r$. As long as there exists a coordinate $i \in [d]- L_r$ where a  $(1 - \delta)$ fraction of points in $F_r$ has value $\?$ on the coordinate $i$, then set $L_r:= L_r \cup \{i\}$ and $F_r$ as the set of points of $F_r$ with value $\?$ on the coordinate $i$. This process has to stop after $\Delta $ steps as $P$ consists of $\Delta$-points. This means that, at the end,  $|F_r| \geq (1- \delta)^{\Delta}|R'_r|$, which ends the proof of the claim. For $i \in L_r$, let $R^i_r = \pd(R'_r, i)$ and $p_i = |R^i_r|/|R'_r|$. The previous discussion shows that $p_i \geq \frac{\delta}{2}$ for all $i \not \in L_r$. 

Intuitively, $L_r$ correspond to the set of coordinates such that, if we sample in $R'_r$, then we might not get a good estimate for $c'_r$ on these coordinates. Suppose now we pick uniformly at random an element $x$ of $R - B$. Because $|F_r| \geq \frac{|R'_r|}{2} \geq \frac{q'|R- B|}{4k}$, with probability at least $\frac{q'}{4k}$, $x \in F_r$. Assume from now on this is the case, and let $J_r = \dom(x)$. Note that $L_r \subseteq [d] - J_r$ and $|J_r|\ge d - \Delta$. 

Let $t  = \frac{8}{q\delta}$. From the choice of $q$ and because $p_i\ge \frac{\delta}{2}$ one can show that $exp(-tp_i^2/4) \leq \frac{2}{tp_i}$. Consider $X = \{x_1, \dots, x_t\}$,  a (multi)set of $t$ elements in $R-B$ obtained by doing $t$ independent and uniform draws. With probability at least $( \frac{q'}{2k})^{t}$, all these points belong to $R'_r$. From now on we assume this is the case, and all the probabilities computed will be conditioned by that fact. Note that the set $\{x_1, \dots, x_t\}$ then follows exactly the same distribution as the one obtained by doing $t$ independent and uniform draws in $R'_r$. Let $J'_r$ be the set of coordinates $e$ of $J_r$ for which there exists an element $x_s \in X$ such that $(x_s)_e \ne\ \?$, note that $|J_r - J'_r| \leq \Delta$. Let $u' = c^{J'_r}(X)$. For every $i \in J'_r$, denote by $X_i$ the random variable counting the number of the points $x_j$ with $j \in [t]$ such that $(x_j)_i \ne\ \?$. Note that $X_i$ follows the binomial distribution with parameters $t$ and $p_i$. By using standard tail bounds of the binomial distribution (see, e.g., Theorem 1 of \cite{doi:10.1080/01621459.1963.10500830}) we can show the following claim: 

\begin{claim}
    For every $i \in J'_r$, $Pr[X_i \leq tp_i/2] \leq exp(-tp_i^2/4)$.
\end{claim}

Moreover, if we condition by the event that $X_i = p$, then the distribution followed by the $p$ values of $(x_s)_i$ is exactly the same as the one obtained by doing $p$ uniform and random draws among all the $v_i$ for $v \in \pd(R'_r, i)$. This implies the following result. 

\begin{claim}\label{clm:proba1}
    For every $i \in J'_r$, we have $\mathbb{E}\left(|(c'_r)_i - (u')_i|^2\right) \leq \frac{\sum_{x \in R^i_r}|x_i - (c'_r)_i|^2}{|R^i_r|} \cdot \frac{4}{tp_i}$.
\end{claim}

\begin{cProof}
    By the definition of a partial clustering and an extension, and because $I_r$ is empty, we know that $(c'_r)_i$ is equal to the average of $s_i$ over all elements $s$ of $R^i_r$. By applying Lemma \ref{lem:variance}, we get that, for every $s \leq t$, \[\mathbb{E} \left( |(c'_r)_i - u'_i|^2 | X_i = s \right) \leq \frac{\sum_{x \in R^i_r}|x_i - (c_r)_i|^2}{|R^i_r|s}. \] 
    This means that: 
    \begin{align*}
        \mathbb{E} \left( |(c'_r)_i - u'_i|^2\right) &=  \mathbb{E} \left( |(c'_r)_i - u'_i|^2 \mid X_i \leq tp_i/2 \right) Pr[X_i \leq tp_i/2] \\ &+ \mathbb{E} \left( |(c'_r)_i - u'_i|^2 \mid X_i > tp_i/2 \right) Pr[X_i > tp_i/2] \\
        &\leq \frac{\sum_{x \in R^i_r}|x_i - (c_r)_i|^2}{|R^i_r|} \cdot Pr[X_i \leq tp_i/2]\\
        &+ \frac{\sum_{x \in R^i_r}|x_i - (c_r)_i|^2}{|R^i_r|tp_i/2}  \\
        & \leq \frac{\sum_{x \in R^i_r}|x_i - (c_r)_i|^2}{|R^i_r|}\left( exp(-tp_i^2/4)+ \frac{2}{tp_i}\right).
    \end{align*}
    Which ends the proof as we chose $q$ such that $exp(-tp_i^2/4) \leq \frac{2}{tp_i}$. 
\end{cProof}

Now for every index $a \in J_r - J'_r$, consider $A_a = \{ x^a_1, \dots, x^a_t \}$ a new (multi)set of $t$ elements in $R-B$ obtained by doing uniform and independent draws. With probability at least $(\frac{q'\delta}{4k})^t$,
all these elements belong to $R^a_r$. From now on, let us assume that it is the case for every $a \in J_r - J'_r$. Setting $u'_a = c^a(A_a)$, an analogous proof as the one of Claim \ref{clm:proba1} would yield: 

\begin{claim}\label{clm:proba2}
    For every $a \in (J_r - J'_r)$, we have $\mathbb{E}\left(|(c'_r)_a - u'_a|^2\right) \leq \frac{\sum_{x \in R^a_r}|x_a - (c'_r)_i|^2}{|R^a_r|} \cdot \frac{4}{tp_a} $.
\end{claim}

By summing over all coordinates of $J_r$, we obtain the following result. 

\begin{claim}\label{cl:proba_1}
    With probability at least $1/2$, 
    $f^{J_r}_2(R'_r, u') \leq ( 1 + \frac{8}{tp_i}) (f^{J_1}_2(R_r, c'_r))$.
\end{claim}

\begin{cProof}
    Indeed, $\mathbb{E}(f^{J_r}_2(R'_r, u') - f^{J_r}_2(R_r, c'_r) ) =  \sum_{i \in J_r}|R^i_r| \cdot \mathbb{E}(|(c'_r)_i - u'_i|^2)$ by Lemma \ref{lem:approx_center_1}, which is smaller than $\frac{4}{tp_i} (f^{J_r}_2(R_r, c'_1)) $ by the previous claims. Markov's inequality allows us to conclude. 
\end{cProof}

Therefore, by choosing the following at random: 
\begin{itemize}
    \item an index $r \in [k]$ such that $|R'_r -\fd(R'_r, I_r)| \geq \frac{q'|R-B|}{2k}$,
    \item an element $x \in R-B$,
    \item $t$ elements $x_1, \dots, x_t$ in $R-B$, and 
    \item $t$ elements $x^a_1, \dots, x^a_1$ in $R-B$ for each $a \in J_r - J'_r$ 
\end{itemize}
we find, with probability at least $\frac{1}{k} \cdot \frac{q'}{4k} \cdot( \frac{q'}{2k})^{(t-1)}\cdot (\frac{q'\delta}{4k})^{t\Delta} \cdot 1/2  \geq \frac{(q'\delta)^{(\Delta+1) t}}{(4k)^{(\Delta+1) t +1}}=\frac{1}{2^{\Oh(\frac{\Delta^3k}{q}\log\frac{1}{q})}}$, a point $u_1 \in \mathbb{H}^d$ such that $f^{J_r}_2(R'_r, u') \leq ( 1 + \frac{8}{tp_i}) (f^{J_r}_2(R_r, c'_r))$. Now consider the set of points $(c_i)_{i \in [k]}$ defined as follows: if $i \not = r$, then $c_i = c'_i$ and $c_r$ is the point obtained from $c'_r$ by setting $(c_r)_j = u'_j$ on all the coordinates of $J_r$. We have that: 
\begin{align*}
\val((P'_i)_{i \in [k]}, (c_i)_{i \in [k]} ) &\leq \val((P'_i)_{i \in [k]}, (c'_i)_{i \in [k]} ) + \frac{8}{tp_i} (f^{J_1}_2(R_r, c'_r))\\
&\leq (1+q)^2OPT(\mathcal{P})\\
&\leq (1 + \alpha)OPT(\mathcal{P}).
\end{align*}

Therefore, it means that the partial clustering $\mathcal{P}'$ obtained from $\mathcal{P}$ by setting $n_r = 1$, $u_r = u'$ and $I_r = J_r$ satisfies the property of the lemma. 
Indeed, the partition $(P'_i)_{i \in [k]}$ is an extension of $\mathcal{P}'$ with value at most $\val\left((P'_i)_{i \in [k]}, (c_i)_{i \in [k]} \right)$ and $|J_r| \geq d - \Delta$.

\paragraph{Case 2: $|I_r| \geq d- \Delta$.}
Let $S  = [d] -I_r$ and note that, by the definition of $\fd(R'_r, I_r)$, for every element $y \in R'_r -\fd(R'_r, I_r)$, there is an index $ j \in S$ such that $j \in \dom(y)$. In particular it means that there exists an index $j \in S$ such that $j \in \dom(y)$ for at least $|R'_r -\fd(R'_r, I_r)|/\Delta$ of the elements of $R'_r -\fd(R'_r, I_r)$. Because $|S| \leq \Delta$, by picking a random element of $S$, with probability at least $1/\Delta$, we can assume that we know this index $j$. Our main goal now will be to guess $(c'_r)_j$. Let $t' =\frac{2}{q}$ and suppose that $s_1, \dots, s_{t'}$ is a (multi)set of elements of $R-B$ obtained by $t'$ uniform and independent draws. With probability at least $(\frac{q'}{2k\Delta})^{t'}$, all the $s_i$'s belong to $\pd(R'_r,j)$. From now on, let us assume that this is the case. Note that in this case, the random set $(s_1, \dots, s_{t'})$ follows the same distribution as one obtained by doing $t$ uniform and independent draws in $\pd(R'_r,j)$. Let $a_j = \sum_{i \in [t]}(s_i)_j/t'$, a proof similar to the one of Claim~\ref{cl:proba_1} gives the following claim: 

\begin{claim}
With probability at least $1/2$, $f^{j}_2(R'_r, a_j) \leq ( 1 + q) (f^{j}_2(R'_r, c'_r))$.
\end{claim} 

In that case, the partial clustering $\mathcal{P'}$ obtained from $\mathcal{P}$ by setting $(u_r)_{j} = a_j$ and $n_r := n_r +1$ satisfies the desired properties. 

By considering both cases, we obtain an algorithm running in time $\Oh(knd)$ and succeeding with probability at least the probability that the algorithm of Lemma \ref{lem:partial_assignment} succeeds times the minimum of
$\frac{(q'\delta)^{(\Delta+1) t}}{(4k)^{(\Delta+1) t +1}}$ and $\frac{1}{\Delta}\cdot(\frac{q'}{2k\Delta})^{t'}$, which is at least 
\[\frac{1}{2^{\Oh(\frac{\Delta^3k}{q}\log\frac{1}{q})}(\log n)^k}.\]

Note that if $\alpha$ is sufficiently small,
then this is $\frac{1}{2^{\Oh(\frac{\Delta^3k}{\alpha}\log\frac{1}{\alpha})}(\log n)^k}$ and else it is $\frac{1}{2^{\Oh({\Delta^6k}\log\Delta)}(\log n)^k}$, finishing the proof.
\end{proof}

Finally, by applying Lemma \ref{lem:algorithm_step} at most $k(\Delta +1)$ times we obtain Theorem \ref{th:kmean_FPTAS}.

\mainTheorem*
\begin{proof}%[Proof of Theorem \ref{th:kmean_FPTAS}]
    Let $P$ be a instance of the {\sc $k$-means clustering} problem consisting of $\Delta$-points.  
    Fix $\alpha = ((1 + \epsilon)^{1/k(\Delta + 1)} -1 )$, note that $\alpha \ge \frac{\epsilon}{3k(\Delta+1)}$, and let $\mathcal{P} = \{ (n_i)_{i \in [k]}, (H_i)_{i \in [k]}, (I_i)_{i \in [k]}, (u_i)_{i \in [k]}  \}$ be the partial clustering such that for each $i\in [k]$, $n_i = 0$, $H_i = \emptyset$, $I_i = \emptyset$ and $u_i$ is the point of $\mathbb{H}^d$ with only \?. Note that $OPT(\mathcal{P})$ is equal to the optimal value of the instance. 

    The algorithm consists of applying inductively Lemma \ref{lem:algorithm_step} with the constant $\alpha$ and Observation~\ref{obs:1} until $\bigcup_{i \in [k]} (H_i) = P$. Since $\sum_{i \in [k]} n_i$ increases in every application of Lemma~ \ref{lem:algorithm_step}, we get by Observation \ref{obs:size_ni} that this process stops after at most $k(\Delta +1)$ steps. The probability that all the steps succeed is at least $(\frac{g(\alpha, k, \Delta)}{\log(n)^k})^{k(\Delta +1)}$, where $g(\alpha, k, \Delta)=\min\{\frac{1}{2^{\Oh(\frac{\Delta^3k}{\alpha}\log\frac{1}{\alpha})}}, \frac{1}{2^{\Oh({\Delta^6k}\log\Delta)}}\}$, and if it holds then the partial clustering $\mathcal{P'} = \{ (n'_i)_{i \in [k]}, (H'_i)_{i \in [k]}, (I'_i)_{i \in [k]}, (u'_i)_{i \in [k]}  \}$ obtained at the end is such that $OPT(\mathcal{P'} ) \leq (1+ \alpha)^{k(\Delta +1)}$. Since $\bigcup_{i \in [k]} (H_i) = P$, this gives us indeed a $(1+\epsilon)$ approximation. Note that we can obtain a center for $H_i$ simply by either taking $u_i$ or computing $c(H_i)$. The running time is $\Oh(k^2\Delta nd)$ and the probability of success is at least $(\frac{g(\alpha, k, \Delta)}{\log(n)^k})^{k(\Delta +1)}$, which means that running the previous algorithm $\Oh((\frac{\log(n)^k}{g(\alpha, k, \Delta)})^{k(\Delta +1)})$
    %$\Oh(\log(n)^{k^2(\Delta +1)})$ 
    times allows us to find the approximation with constant probability. 
    Finally, it is well-known that for any constant $C$, $\log(n)^C = n + C^{\Oh(C)}$\ee{Maybe add a citation? I can find one later..}{} which gives the total running time of: 
    \[
    \max\{{2^{\Oh(\frac{\Delta^3k}{\alpha}\log\frac{1}{\alpha})}}, {2^{\Oh({\Delta^6k}\log\Delta)}}\}^{k(\Delta +1)}2^{\Oh(k^2\Delta)\log(k\Delta)}k^2\Delta n^2d
    \]
    which can be simplified to 
    \[
     \max\{{2^{\Oh(\frac{\Delta^5k^3}{\epsilon}\log\frac{k\Delta}{\epsilon})}}, {2^{\Oh({\Delta^7k^2}\log(k\Delta))}}\} n^2d,
    \]
    finishing the proof. 
\end{proof}

\section{Concluding Remarks}

In this paper we gave the first PTAS for $k$-means clustering of $\Delta$-points when $\Delta>1$ running in time ${2^{O(\frac{\Delta^7k^3}{\epsilon}\log\frac{k\Delta}{\epsilon})}} n^2d$ based on iteratively sampling points to discover new coordinates of some center. We believe that the study of clustering problems of $\Delta$-points is an interesting research direction and there is still a lot to be discovered. 
We conclude with concrete open questions. 

Arguable, the most popular clustering objectives are $k$-center, $k$-means, and $k$-median. 
For $k$-center clustering of $\Delta$-points a PTAS was obtained by 
  Lee and  Schulman in \cite{LeeS13}. However, for $k$-median clustering of $\Delta$-points, the question whether the problem admits  a PTAS,  remains open. 
 We would like to remark here that the algorithm of Kumar et al.~\cite{KumarSS10} for clustering of points in $\mathbb{R}^d$ works not only for $k$-means, but also for $k$-medial clustering.
 
 Since $\Delta$-points are basically $\Delta$-dimensional axis-parallel subspaces, another interesting question would be whether it is possible to extend our result to clustering of arbitrary $\Delta$-dimensional affine subspaces in $\mathbb{R}^d$. This is a very natural computational geometry problem which complexity, to the best of our knowledge, is widely open. 
 
Following the coreset construction for $k$-means clustering of lines by Marom and Feldman~\cite{MaromF19}, it is a natural open question whether it is possible do design a coresets of small size for clustering of $\Delta$-points for $\Delta > 1$. For lines, the size of  coreset of 
Marom and Feldman is
$dk^{\Oh(k)} \log n/\epsilon^2$. In particular, whether $\log n$ can be removed even for $\Delta=1$,  is open.

Finally, we do not know  how tight is the running time of our algorithm. While it is plausible to suggest that  the dependency in $k$ and $\Delta$ is not optimal, to design a faster algorithm we need new ideas. It ls also an interesting open question whether one can improve the dependency on $n$ from quadratic to linear.

\end{document}